%% file: transHDA.tex
\documentclass{lmcs}
\pdfoutput=1

\usepackage{lastpage}
\lmcsdoi{15}{3}{28}
\lmcsheading{}{\pageref{LastPage}}{}{}%
{Feb.~13,~2018}{Sep.~06,~2019}{}

\usepackage{amssymb}
\usepackage{stmaryrd}
\usepackage{mathrsfs}
\usepackage{subfigure}
\usepackage[all,cmtip]{xy}
\usepackage{mathtools}
\usetikzlibrary{automata,positioning}

\def\N{{\mathbb{N}}}
\def\Z{{\mathbb{Z}}}
\def\A{{\mathcal{A}}}
\def\B{{\mathcal{B}}}
\def\C{{\mathcal{C}}}
\def\T{{\mathcal{T}}}
\def\S{{\mathcal{S}}}
\def\P{{\mathcal{P}}}

\begin{document}

\title {Higher-dimensional automata modeling shared-variable systems}

\author{Thomas Kahl}

\address{Centro de Matem\'atica,
	Universidade do Minho, Campus de Gualtar,
	4710-057 Braga,
	Portugal
}

\thanks{This research was financed by Portuguese funds through FCT - Funda\c c\~ao para a Ci\^encia e a Tecnologia (project UID/MAT/00013/2013).}

\email{kahl@math.uminho.pt}


\begin{abstract}
The purpose of this paper is to provide a construction to model shared-variable systems using higher-dimensional automata which is compositional in the sense that the parallel composition of completely independent systems is modeled by the standard tensor product of HDAs and nondeterministic choice is represented by the coproduct. 	
\end{abstract}

\maketitle 

\section{Introduction}

Higher-dimensional automata provide a rich model for concurrent systems. A higher-dimensional automaton (HDA) is a precubical set with an initial state, a set of final states, and a labeling on 1-cubes such that opposite edges of 2-cubes have the same label. An HDA is thus a transition system (or an ordinary automaton) with incorporated squares and cubes of higher dimensions. An $n$-cube in an HDA indicates that the $n$ actions starting at its origin are independent in the sense that they may be executed in any order, or even simultaneously, without any observable difference. The concept of higher-dimensional automaton goes back to Pratt \cite{Pratt}. The notion defined here is due to van Glabbeek  \cite{vanGlabbeek}.

A natural approach to constructing an HDA from a transition system, which has been described in various places in the literature (see e.g. \cite{GaucherProcess, GaucherCombinatorics, vanGlabbeek, GoubaultLabCubATS, GoubaultMimram}), is to fill in empty squares and higher-dimensional cubes. In order to realize this strategy concretely, one needs to specify not only which cubes to fill in but also how to fill in a cube. To make this more precise, suppose that we consider the actions $a$ and $b$ in the empty square depicted in Figure \ref{fig1a} 
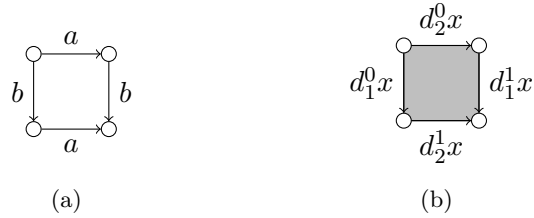
\begin{figure}[t]{}	
	\centering
	\begin{subfigure}[]
	{ 
		
		\begin{tikzpicture}[initial text={},on grid]  
		\path(0,-5.75mm);
					    
		\node[state,minimum size=0pt,inner sep =2pt,fill=white] (p_0)   {}; 
		
		\node[state,minimum size=0pt,inner sep =2pt,fill=white] (p_2) [right=of p_0,xshift=0cm] {};
		
		\node[state,minimum size=0pt,inner sep =2pt,fill=white,initial where=above,initial distance=0.2cm] [above=of p_0, yshift=0cm] (p_3)   {};

		\node[state,minimum size=0pt,inner sep =2pt,fill=white] (p_5) [right=of p_3,xshift=0cm] {}; 
		
		\path[->] 
		(p_0) edge[below] node {$a$} (p_2)
		(p_3) edge[above]  node {$a$} (p_5)
		(p_3) edge[left]  node {$b$} (p_0)
		(p_5) edge[right]  node {$b$} (p_2);

		\end{tikzpicture}
		\label{fig1a} 
	} 
\end{subfigure} 
	\hspace{2cm}
\begin{subfigure}[] 
	{
		\begin{tikzpicture}[initial text={},on grid]

		\path[draw, fill=lightgray] (0,0)--(1,0)--(1,1)--(0,1)--cycle;

		\node[state,minimum size=0pt,inner sep =2pt,fill=white] (q_0)   {}; 
		
		\node[state,minimum size=0pt,inner sep =2pt,fill=white] (q_2) [right=of q_0,xshift=0cm] {};

		\node[state,minimum size=0pt,inner sep =2pt,fill=white] [above=of q_0, yshift=0cm] (q_3)   {};

		\node[state,minimum size=0pt,inner sep =2pt,fill=white] (q_5) [right=of q_3,xshift=0cm] {}; 
		
		\path[->] 
		(q_0) edge[below] node {$d^1_2x$} (q_2)
		(q_3) edge[above]  node {$d^0_2x$} (q_5)
		(q_3) edge[left]  node {$d^0_1x$} (q_0)
		(q_5) edge[right]  node {$d^1_1x$} (q_2);

		\end{tikzpicture}
		\label{fig1b}
	}  
\end{subfigure} 	
	\caption{An empty independence square and a 2-cube $x$ with its boundary edges} 
\end{figure} 
to be independent and intend to fill in the square. Then we have to adjoin a 2-cube $x$ to the square and define its boundary edges $d^0_1x$, $d^0_2x$, $d^1_1x$, and $d^1_2x$ (see Figure \ref{fig1b} for a picture). There are two possibilities for this because we can choose which of the front faces $d^0_1x$ and $d^0_2x$ has label $a$ and which one has label $b$. To fill in the square, we may use either or both possibilities and attach one or two (or even more) 2-cubes.

Since there is no canonical way to attach a 2-cube to an empty square, most methods to construct HDAs from transition systems are based on the symmetric approach to fill in squares always in both possible ways, and it has been shown that this is well suited to many purposes  \cite{GaucherCombinatorics, vanGlabbeek, GoubaultLabCubATS, GoubaultMimram}. A drawback of this symmetric way to proceed is, however, that it leads to a certain redundancy in the representation of the independence relation: the independence of $n$ actions is represented by not just one but a significantly larger number of $n$-cubes. Moreover, the symmetric construction does not explain very well the intuition that the usual tensor product of HDAs models the parallel composition of completely independent concurrent systems.

In this paper we focus on transition systems and HDAs modeling finite systems of asynchronously executing processes which communicate through shared variables. Given a transition system representing the state space of such a concurrent system, a canonical choice between the two possibilities to fill in an empty independence square is possible if process IDs are taken into account: we may attach a 2-cube always in such a way that the process behind the action of the first front edge has a higher process ID than the one associated with the second front edge. Based on this attaching rule for 2-cubes, one can set up a filling procedure to construct HDAs from transition system models of  shared-variable programs. In an HDA constructed using this method the independence of $n$ actions at a state is represented by a unique $n\mbox{-}$cube (at least if the system under consideration consists of deterministic processes). Thus the independence relation is described in a nonredundant fashion. Moreover, the construction of the HDA is compositional in the sense that the tensor product of HDAs models independent parallel composition and the coproduct represents nondeterministic choice. The main objective of this paper is to establish these compositionality properties of the construction.

Process IDs are a means to order pairs of independent actions. In other words, they turn independence relations asymmetric. It turns out that instead of using process IDs, we may more abstractly work with asymmetric or even plain relations. To do so, we introduce $\ltimes$-transition systems, which are transition systems with a relation on labels that does not need to have any properties. The state space of a shared-variable system can be modeled as a $\ltimes$-transition system with an asymmetric relation, and we shall provide more details for systems of program graphs \cite{BaierKatoen, MannaPnueli} later on in the paper. 

The fundamental material on $\ltimes $-transition systems is presented in Section \ref{ltimes}. The earlier Section \ref{SecPrel} contains preliminaries on precubical sets and higher-dimensional automata. For every $\ltimes $-transition system, we construct an HDA model by filling in empty cubes. The construction is functorial and has a left adjoint. HDA models are the subject of Section \ref{SecHDAmod}. In Section \ref{SecParallel} we define an interleaving operator for $\ltimes$-transition systems and show that it models the parallel composition of shared-variable systems without common variables. Our main result is Theorem \ref{mainresult}, which states that the HDA model of the interleaving of two $\ltimes$-transition systems is the tensor product of their HDA models. In Section \ref{SecChoice} we turn our attention to nondetermistic choice and show that our HDA model construction preserves coproducts.  The final section contains a few concluding remarks.

\section{Precubical sets and higher-dimensional automata} \label{SecPrel}

This section presents some well-known material about precubical sets and higher-dimensional automata.

\subsection{Precubical sets} \label{precubs}

A \emph{precubical set} is a graded set $P = (P_n)_{n \geq 0}$ with  \emph{boundary} or \emph{face operators} $d^k_i\colon P_n \to P_{n-1}$ $(n>0,\;k= 0,1,\; i = 1, \dots, n)$ satisfying the relations $d^k_i\circ d^l_{j}= d^l_{j-1}\circ d^k_i$ $(k,l = 0,1,\; i<j)$. If $x\in P_n$, we say that $x$ is of  \emph{degree} $n$ and write $\deg(x) = n$. The elements of degree $n$ are called the \emph{$n$-cubes} of $P$. The elements of degree $0$ are also called the \emph{vertices} of $P$, and the $1$-cubes are also called the \emph{edges} of $P$. A \emph{precubical subset} of a precubical set $P$ is a graded subset of $P$ that is stable under the boundary operators. A \emph{morphism} of precubical sets is a morphism of graded sets that is compatible with the boundary operators. The category of precubical sets will be denoted by ${{\mathsf{PCS}}}$. 

The category of precubical sets can be seen as the presheaf category of functors $\Box^{\mbox{\tiny op}} \to {\mathsf{Set}}$ where $\Box$ is the small subcategory of the category of topological spaces whose objects are the standard $n$-cubes $[0,1]^n$ $(n \geq 0)$ and whose nonidentity morphisms are composites of the maps $\delta^k_i\colon [0,1]^n\to [0,1]^{n+1}$ ($k \in \{0,1\}$, $n \geq 0$, $i \in  \{1, \dots, n+1\}$) given by $\delta_i^k(u_1,\dots, u_n)= (u_1,\dots, u_{i-1},k,u_i \dots, u_n)$.

\subsection{Tensor product of precubical sets}

Given two precubical sets $P$ and $Q$, the \emph{tensor product} $P\otimes Q$ is the precubical set defined as follows: The underlying graded set is given by \[(P\otimes Q)_n = \coprod \limits_{p+q = n} P_p\times Q_q.\] For an $n$-cube $(x,y) \in P_p\times Q_q \subseteq (P\otimes Q)_n$, the boundary operators are defined by
\[d_i^k(x,y) = \left\{ \begin{array}{ll} (d_i^kx,y), & 1\leq i \leq p,\\
(x,d_{i-p}^ky), & p < i \leq n.
\end{array}\right.\] 
With respect to this tensor product, the category ${{\mathsf{PCS}}}$ is a monoidal category.

\subsection{Precubical intervals} 
\begin{sloppypar}
Let $k$ and $l$ be two integers such that $k \leq l$. The \emph{precubical interval}  $\llbracket k,l \rrbracket$ is the  precubical set defined by $\llbracket k,l \rrbracket_0 = \{k,\dots , l\}$, $\llbracket k,l \rrbracket_1 =  \{{[k,k+ 1]}, \dots , {[l- 1,l]}\}$, $d_1^0[j-1,j] = j-1$, $d_1^1[j-1,j] = j$, and $\llbracket k,l \rrbracket_{n} = \emptyset$ for $n > 1$.
\end{sloppypar}

\subsection{Precubical cubes}
	The \emph{precubical $n\mbox{-}$cube} is the $n$-fold tensor product ${\llbracket 0,1\rrbracket^ {\otimes n}}$. Here, we use the convention that ${\llbracket 0,1\rrbracket^ {\otimes 0}}$ is the precubical set $\llbracket 0, 0 \rrbracket = \{0\}$. The only element of degree $n$ in $\llbracket 0,1\rrbracket^ {\otimes n}$ will be denoted by $\iota_n$. We thus have $\iota_0 = 0$ and $\iota_n = (\underbracket[0.5pt]{ [0,1] ,\dots , [ 0,1]}_{n\; {\text{times}}})$ for $n>0$. Given an element $x$ of degree $n$ of a precubical set $P$, there exists a unique morphism of precubical sets $x_{\sharp}\colon \llbracket 0,1\rrbracket ^{\otimes n}\to P$ that sends $\iota_n$ to $x$. Indeed, 
	by the Yoneda lemma, there exist unique morphisms of precubical sets $f\colon \Box(-,[0,1]^n) \to P$ and $g\colon \Box(-,[0,1]^n) \to \llbracket 0,1\rrbracket ^{\otimes n}$ such that $f(id_{[0,1]^n}) = x$ and $g(id_{[0,1]^n}) = \iota_n$. The map $g$ is an isomorphism, and $x_{\sharp} = f\circ g^{-1}$.

\subsection{Truncated precubical sets}\label{truncation}
Let $n \geq 0$ be an integer. An \emph{$n$-truncated} precubical set is a precubical set $P$ such that $P_m = \emptyset$ for all $m > n$. The full subcategory of ${{\mathsf{PCS}}}$ given by the $n$-truncated precubical sets is denoted by ${{\mathsf{PCS}}_n}$. The \emph{$n\mbox{-}$skeleton} of a precubical set $P$ is the $n$-truncated precubical set $P_{\leq n}$ defined by $(P_{\leq n})_m = P_m$ $(m\leq n)$. The association $P \mapsto P_{\leq n}$ defines the 
\emph{$n$-truncation functor} $T_n \colon {{\mathsf{PCS}}} \to {{\mathsf{PCS}}_n}$. The functor $T_n$ has both a left and a right adjoint (cf. e.g. \cite[6.2.6]{Riehl}). The left adjoint is the inclusion ${{\mathsf{PCS}}_n} \hookrightarrow {{\mathsf{PCS}}}$. The right adjoint, which we will denote by $\varGamma_n$, sends an $n$-truncated precubical set $P$ to its \emph{$n$-coskeleton}, i.e., the unique (up to isomorphism)  maximal precubical set such that the $n$-skeleton is $P$ and no two different cubes of degree $>n$ have the same boundary. We remark that natural bijections $(\varGamma_nP)_m \to \mathsf{PCS}(\llbracket 0,1\rrbracket^{\otimes m}_{\leq n},P)$ are given by $x \mapsto (x_{\sharp})_{\leq n}$.

\subsection{Paths} \label{paths}
A \emph{path} of \emph{length $l$} $(l \geq 0)$ from a vertex $v$ of a precubical set $P$ to a vertex $w$ is a morphism of precubical sets $\omega \colon \llbracket 0,l \rrbracket \to P$ such that $\omega(0) = v$ and $\omega(l) = w$. The set of paths in $P$ is denoted by $P^{\mathbb I}$. The \emph{concatenation} $\omega \cdot\nu$  of two paths $\omega \colon \llbracket 0,k \rrbracket \to P$ and $\nu \colon \llbracket 0,l \rrbracket \to P$ with $\omega (k) = \nu (0)$ is defined in the obvious way. Note that every path in $P$ can be written as a finite concatenation of paths of the form $x_{\sharp}$ where $x \in P_{\leq 1}$.  

Given two precubical sets $P^1$ and $P^2$, there are unique maps  $\pi^i\colon (P^1\otimes P^2)^{\mathbb I} \to (P^i)^{\mathbb I}$  $(i \in \{1,2\})$ that are compatible with concatenation and satisfy $\pi^i((x^1,x^2)_{\sharp}) = x^i_{\sharp}$ for all $(x^1,x^2) \in (P^1\otimes P^2)_{\leq 1}$. If $\omega \in (P^1\otimes P^2)^{\mathbb I}$ is a path from $(v^1,v^2)$ to $(w^1,w^2)$, then $\pi^i(\omega)$ is a path in $P^i$ from $v^i$ to $w^i$. The map $(\pi^1,\pi^2) \colon (P^1\otimes P^2)^{\mathbb I} \to (P^1)^{\mathbb I} \times (P^2)^{\mathbb I}$ is not bijective in general but has a right inverse. If for $i \in \{1,2\}$, $\omega^i$ is a path in $P^i$ from $v^i$ to $w^i$, then this right inverse sends the pair $(\omega^1, \omega^2)$ to a path from $(v^1,v^2)$ to $(w^1,w^2)$.

\subsection{Higher-dimensional automata} \label{HDAdef}

A \emph{higher-di\-mensional automaton} (HDA) is a tuple \[\A = (P_{\A},I_{\A},F_{\A},\Sigma_{\A}, \lambda_{\A})\] where  $P_{\A}$ is a precubical set, ${I_{\A} \in (P_{\A})_0}$ is a vertex, called the \emph{initial state}, ${F_{\A} \subseteq (P_{\A})_0}$ is a possibly empty set of \emph{final} or \emph{accepting  states}, $\Sigma_{\A}$ is a set of \emph{labels}, and $\lambda_\A \colon (P_{\A})_1 \to \Sigma_{\A}$ is a map, called the \emph{labeling function}, such that  $\lambda_{\A} (d_i^0x) = \lambda_{\A} (d_i^1x)$ for all $x \in (P_{\A})_2$ and $i \in \{1,2\}$ \cite{vanGlabbeek}. A \emph{morphism} from an HDA $\A$ to an HDA $\B$ consists of a morphism of precubical sets  $f\colon P_\A \to P_\B$ and a map $\sigma \colon \Sigma_{\A} \to \Sigma_{\B}$ such that $f(I_{\A}) = I_{\B}$, $f(F_{\A}) \subseteq F_{\B}$, and  $\lambda_{\B}(f(x)) =  \sigma(\lambda_{\A}(x))$ for all $x \in (P_\A)_1$. The category of HDAs is denoted by $\sf HDA$.

Given two HDAs  $\A$ and $\B$, we say that $\B$ is a \emph{subautomaton} of $\A$ and write $\B \subseteq \A$ if $P_\B$ is a precubical subset of $P_\A$, $\Sigma_B \subseteq \Sigma_A$, $I_{\B} = I_{\A}$, $F_{\B}\subseteq F_{\A}$, and $\lambda_{\B}(x) = \lambda_{\A}(x)$ for all $x\in {(P_\B)_1}$.

An \emph{$n$-truncated} HDA is an HDA $\A$ such that $P_\A$ is an $n$-truncated precubical set. The \emph{$n$-skeleton} of an HDA $\A$ is the $n$-truncated HDA $\A_{\leq n}$ defined by $P_{\A_{\leq n}} = (P_{\A})_{\leq n}$, $I_{\A_{\leq n}} = I_{\A}$, $F_{\A_{\leq n}} = F_{\A}$, $\Sigma_{\A_{\leq n}} = \Sigma_{\A}$, and $\lambda_{\A_{\leq n}} = \lambda_{\A}|_{(P_{\A_{\leq n}})_1}$.

\section{\texorpdfstring{$\ltimes$}{}-Transition systems} \label{ltimes}

We introduce $\ltimes$-transition systems, which are transition systems with a relation on the set of labels. We make precise how shared-variable systems can be modeled by $\ltimes\mbox{-}$transition systems and relate $\ltimes\mbox{-}$transition systems to 2-truncated higher-dimensional automata.

\subsection{Deterministic and extensional HDAs} An HDA is said to be \emph{deterministic} if its $1$-skeleton is deterministic in the sense of automata theory, i.e., if any two edges with the same label and the same initial vertex are equal. The following weaker condition on HDAs is essential in the sequel (see Section \ref{secHDAinter}): we say that an HDA is \emph{extensional} if any two edges with the same label and the same initial and final vertices are equal. The full subcategory of $\mathsf{HDA}$ given by the extensional HDAs will be denoted by $\mathsf{HDA_{ext}}$.  

\subsection{Transition systems and \texorpdfstring{$\ltimes$}{}-transition systems} \label{transdef}
\begin{sloppypar}
A \emph{transition system} is a $1\mbox{-}$truncated extensional HDA. A \emph{$\ltimes$-transition system} is a tuple 
\[\T = (P_{\T},I_{\T},F_{\T},\Sigma_{\T},\lambda_{\T},\ltimes_\T)\]
where $U(\T) = (P_{\T},I_{\T},F_{\T},\Sigma_{\T},\lambda_{\T})$ is a transition system and $\ltimes_\T$ is a relation on $\Sigma_T$. The transition system $U(\T)$ is called the \emph{underlying transition system}. A \emph{morphism} from a $\ltimes\mbox{-}$transition system $\S$ to a $\ltimes$-transition system $\T$ is a morphism of HDAs $(f,\sigma) \colon U(\S) \to U(\T)$ such that $\sigma$ is compatible with the relations $\ltimes_\S$ and $\ltimes_\T$. The category of $\ltimes$-transition systems will be denoted by $\mathsf{TS_{\ltimes}}$. A $\ltimes\mbox{-}$transition system is called \emph{deterministic} if the underlying transition system is a deterministic HDA. 
\end{sloppypar}
\subsection{Asynchronous transition systems} One may view $\ltimes$-transition systems as generalized asynchronous transition systems in the sense of \cite{WinskelNielsen} (see also \cite{BednarczykThesis, ShieldsConcurrentMach}). These are transition systems with an irreflexive and symmetric relation that has to satisfy certain conditions and represents independence of actions. 

\subsection{Shared-variable systems given by program graphs} 

Let $Var$ be a set of \emph{variables}. The \emph{domain} of a variable $x$, i.e., the set of its possible values, will be denoted by $D_x$. A \emph{program graph} over $Var$ is a tuple \[(V, L, A, T, g, \imath
)\]
the components of which are described as follows (cf. \cite{BaierKatoen, MannaPnueli}):
\begin{itemize}
	\item $V \subseteq Var$ is a set of \emph{declared variables}.  
	\item $L$ is a set of control-flow  \emph{locations}.
	\item $A$ is a set of \emph{actions}. An action is a function $Ev(V) \to Ev(V)$ where  $Ev(W)$ $(W \subseteq Var)$ is the set of \emph{evaluations} of the variables in $W$, i.e.,  \[\quad \quad \quad Ev(W) = \{\eta\colon W \to \cup _{x\in Var} D_x \, |\, \forall \, x \in W : \eta(x) \in D_x\} \cong \prod \limits_{x \in W} D_x. \]
	\item $T\subseteq L \times A  \times L$ is a set of \emph{transitions}. A transition is thus a triple
	\[t = (s_{t}, a_{t},  e_{t})\]
	consisting of a start location $s_{t}$, an action $a_{t}$, and an end location $e_{t}$. 
	\item $g\colon T \to {\mathfrak P}(Ev(V))$ is a function that associates with every transition a \emph{guard condition} on the values of the variables, which is represented by the subset of $Ev(V)$ containing the evaluations for which the condition is true and therefore the transition is enabled. 
	\item $\imath \in L$ is the \emph{initial location} of the program graph.	
\end{itemize}
We say that a program graph is \emph{deterministic} if any two transitions with the same start location and the same action are equal. A \emph{shared-variable system} over $Var$ is a tuple $(\P_1, \dots, \P_n)$  of program graphs over $Var$.

\subsection{The \texorpdfstring{$\ltimes$}{}-transition system model of a shared-variable system}

Consider the shared-variable system $(\P_1, \dots, \P_n)$ over $Var$ given by the program graphs 
\[\P_i = (V_i, L_i, A_i, T_i, g_i, \imath_i), \quad i = 1, \dots, n,\]
and fix an \emph{initial evaluation} $\eta_0 \in Ev(Var)$. 
Define the set of \emph{declared variables} of the system by 
$V = \bigcup \limits_{i=1}^{n} V_i$. Note that $V$ will normally consist of shared variables and not be a disjoint union. 
The \emph{states} of the system  $(\P_1, \dots, \P_n)$ are the elements of the set 
\[Q_0 = L_1 \times \cdots \times L_n \times Ev(V).\] 
The \emph{initial state} of the system is the state
$I  = (\imath_1, \dots , \imath_n, \eta_0|_V)$. The set $Q_0$ is the set of vertices of a 1-truncated precubical set $Q$, which we call the \emph{state graph} of $(\P_1, \dots, \P_n)$. The set of edges of $Q$ is given by 
\[Q_1 = \bigcup_{\substack{i\in \{1,\dots,n\} \\t \in T_i}} L_1 \times \cdots \times  L_{i-1} \times \{t\} \times L_{i+1} \times \cdots \times L_n \times g_i({t}) \times Ev(V \setminus V_i).\]
Before we define the boundary operators of $Q$, let us note that for any two disjoint subsets $W, W' \subseteq V$, we have a bijection $Ev(W \cup W') \to Ev(W) \times Ev(W')$, $\eta \mapsto (\eta|_{W_1}, \eta|_{W_2})$. We denote the inverse bijection by $\ast$. The starting point of an edge \[x = (l_1,\dots, l_{i-1},t,l_{i+1}, \dots, l_{n},\gamma,\eta)\] is the state \[d^0_1x = (l_1,\dots, l_{i-1},s_{t},l_{i+1}, \dots, l_{n},\gamma \ast \eta),\] and its endpoint is the state \[d^1_1x = (l_1,\dots, l_{i-1},e_{t},l_{i+1}, \dots, l_{n},a_{t}(\gamma)\ast \eta).\] 

Let $\T$ be the $\ltimes$-transition system defined as follows:
\begin{itemize}
	\item $P_\T$ is the largest precubical subset of $Q$ the vertices of which are the \emph{reachable} states, i.e., the states to which there exists a path from $I$. 
	\item $I_\T = I$.
	\item $F_\T = \emptyset$.
	\item $\Sigma_\T = \bigcup \limits_{i= 1}^n \{i\}\times A_i$.
	\item  Given an edge $x = (l_1,\dots, l_{i-1},t,l_{i+1}, \dots, l_{n},\gamma,\eta) \in (P_\T)_1$, we say that $i$ is the \emph{process ID} of $x$ and set $\lambda_\T(x) = (i,a_t)$.  
	\item For labels $(i,a), (j,b) \in \Sigma_\T$, we set $(i,a) \ltimes_\T (j,b)$ if $i < j$. 
\end{itemize}
The $\ltimes$-transition system $\T$ is the \emph{$\ltimes$-transition system model} of  $(\P_1, \dots, \P_n)$ with respect to the initial evaluation $\eta_0$. It will be denoted by $\T_{\P_1, \dots, \P_n}$. 

\begin{prop} \label{asymindet}
The relation $\ltimes_{\T_{\P_1, \dots, \P_n}}$ is asymmetric. If each program graph $\P_i$ is deterministic, then so is the $\ltimes$-transition system $\T_{\P_1, \dots, \P_n}$.
\end{prop}

\begin{proof}
It is clear that the relation $\ltimes_{\T_{\P_1, \dots, \P_n}}$ is asymmetric. Suppose that the program graphs $\P_i$ are deterministic. Consider two edges \[x = (l_1,\dots, l_{i-1},t,l_{i+1}, \dots, l_{n},\gamma,\eta) \quad \mbox{and} \quad y = (l'_1,\dots, l'_{j-1},t',l'_{j+1}, \dots, l'_{n},\gamma',\eta')\] such that $d^0_1x = d^0_1y$ and $\lambda_{\T_{\P_1, \dots, \P_n}}(x) = \lambda_{\T_{\P_1, \dots, \P_n}}(y)$. Then $i = j$ and $a_t = a_{t'}$. Consequently, $\gamma = \gamma'$, $\eta = \eta'$, $s_t = s_{t'}$, and $l_r = l'_r$ for all $r$. Moreover, by our hypothesis, $t = t'$. Thus $x=y$.		
\end{proof}

\begin{exa} \label{specex}
	Let $x$ be an integer variable. Set $Var = \{x\}$, and consider the deterministic program graph $\P = (V,L,A,T,g,\imath)$ over $Var$ given by $V = \{x\}$, $L = \{0,1\}$, $A = \{x{\scriptstyle{++}}, x{\scriptstyle{--}}\}$, $T = \{(0,x{\scriptstyle{++}},1), (1,x{\scriptstyle{--}},0)\}$, $g(0,x{\scriptstyle{++}},1) = g(1,x{\scriptstyle{--}},0) = Ev(V) = \{\eta \colon \{x\} \to \Z\}$, $\imath = 0$. We identify $Ev(V) = \Z$ by means of the correspondence $\eta \mapsto \eta(x)$. The actions $x{\scriptstyle{++}}$ and $x{\scriptstyle{--}}$ respectively increment and decrement the value of $x$. Let $\eta_0 = 0$ be the initial evaluation. Consider first the one-process system $(\P)$. The $\ltimes$-transition system $\T_\P$ has two states, namely  $(0,0)$ and $(1,1)$ (recall that we consider reachable states only), and two transitions, one labeled  $(1,x{\scriptstyle{++}})$ from $(0,0)$ to $(1,1)$ and one labeled  $(1,x{\scriptstyle{--}})$ from $(1,1)$ to $(0,0)$. The set $\Sigma_{\T_\P}$ consists of these two labels, and one has $I_{\T_\P} = (0,0)$ and  $\ltimes_{\T_\P} = \emptyset$. Consider now the shared-variable system $(\P_1, \P_2)$ consisting of two copies of $\P$. The underlying transition system of $\T_{\P_1,\P_2}$ is depicted in Figure \ref{fig2}. 
	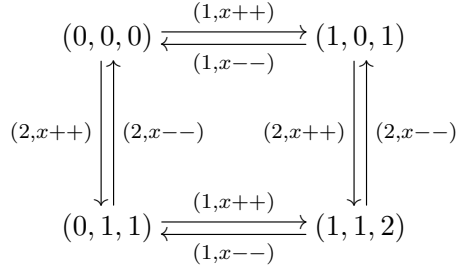
\begin{figure}  
		\centerline{ 
  		\xymatrix{
		(0,0,0) \ar@<.5ex>[rr]^{(1,x{\scriptstyle{++}})} \ar@<-.5ex>[dd]_{(2,x{\scriptstyle{++}})} 
		& & (1,0,1) \ar@<.5ex>[ll]^{(1,x{\scriptstyle{--}})} \ar@<-.5ex>[dd]_{(2,x{\scriptstyle{++}})}\\ 
		& & \\
		(0,1,1) \ar@<.5ex>[rr]^{(1,x{\scriptstyle{++}})} \ar@<-.5ex>[uu]_{(2,x{\scriptstyle{--}})} 
		& & (1,1,2) \ar@<.5ex>[ll]^{(1,x{\scriptstyle{--}})} \ar@<-.5ex>[uu]_{(2,x{\scriptstyle{--}})}
		} 
	} 
	\caption{Transition system of a simple shared-variable system}	\label{fig2}
	\end{figure}
	We have $I_{\T_{\P_1,\P_2}} = (0,0,0)$, $\Sigma_{\T_{\P_1,\P_2}} = \{(1,x{\scriptstyle{++}}),(1,x{\scriptstyle{--}}),(2,x{\scriptstyle{++}}),(2,x{\scriptstyle{--}})\}$, and 
	\begin{align*}
	\ltimes_{\T_{\P_1,\P_2}} &= \{((1,x{\scriptstyle{++}}),(2,x{\scriptstyle{++}})), ((1,x{\scriptstyle{++}}),(2,x{\scriptstyle{--}})),\\
	&\quad \;\;\;
	 ((1,x{\scriptstyle{--}}),(2,x{\scriptstyle{++}})), ((1,x{\scriptstyle{--}}),(2,x{\scriptstyle{--}}))
	\}.
	\end{align*}
\end{exa}

\subsection{Independence squares} An \emph{independence square} in a $\ltimes$-transition system $\T$ is a family of edges  $(x^k_i)_{k\in\{0,1\},i\in\{1,2\}}$ such that 
\begin{enumerate}
	\item for all $k, l \in \{0,1\}$, $d^k_1x^l_2 = d^l_1x^k_1$; and
	\item $\lambda_\T(x^0_2) = \lambda_\T(x^1_2) \ltimes_\T \lambda_\T(x^0_1) = \lambda_\T(x^1_1)$.  
\end{enumerate}
We remark that morphisms of $\ltimes$-transition systems preserve independence squares.

\subsection{\texorpdfstring{$\ltimes$}{}-Transition systems and 2-truncated extensional HDAs} \label{PsiPhi}
 
\begin{sloppypar}
Given a $\ltimes\mbox{-}$transition system $\T$, let $\varPsi(\T)$ be the 2-truncated extensional HDA defined by ${\varPsi (\T)_{\leq 1} = U(\T)}$, $(P_{\varPsi (\T)})_{2} = \{\mbox{independence squares in } \T\}$, and $d^l_j((x^k_i)_{k,i}) = x^l_j$. Given a morphism $(f,\sigma) \in \mathsf{TS_{\ltimes}}(\S,\T)$, consider the morphism of $2\mbox{-}$truncated precubical sets $\psi(f,\sigma) \colon P_{\varPsi (\S)} \to P_{\varPsi (\T)}$ given by ${\psi(f,\sigma)_{\leq 1} = f}$ and $\psi(f,\sigma)((x^k_i)_{k,i}) = (f(x^k_i))_{k,i}$ and define the morphism of 2$\mbox{-}$truncated extensional HDAs ${\varPsi(f,\sigma) \colon \varPsi(\S) \to \varPsi(\T)}$ by $\varPsi(f,\sigma) = (\psi(f,\sigma),\sigma)$. We obtain a functor from $\mathsf{TS}_{\ltimes}$ to the category ${\mathsf{HDA_{ext}}}_2$ of 2$\mbox{-}$truncated extensional HDAs.
\end{sloppypar}

Let $\A$ be a 2-truncated extensional HDA. Define the $\ltimes\mbox{-}$transition system $\varPhi(\A)$ by $U(\varPhi(\A)) = \A_{\leq 1}$ and 
\[
\alpha \ltimes_{\varPhi(\A)} \beta \; \Leftrightarrow \; \exists\, x \in (P_\A)_2 : \lambda_\A(d^0_2x) = \alpha,\,\lambda_\A(d^0_1x) = \beta.
\]
Consider a morphism $(f,\sigma) \in {\mathsf{HDA_{ext}}}_2(\A,\B)$. Then $\sigma$ is compatible with the relations of $\varPhi(\A)$ and $\varPhi(\B)$. Indeed, suppose that $\alpha \ltimes_{\varPhi(\A)} \beta$. Then there exists an element $x \in (P_\A)_2$ such that $\lambda_\A(d^0_2x) = \alpha$ and $\lambda_\A(d^0_1x) = \beta$. Hence \[\lambda_\B(d^0_2f(x)) = \lambda_\B(f(d^0_2x)) = \sigma(\lambda_\A(d^0_2x)) = \sigma(\alpha).\]
Similarly, $\lambda_\B(d^0_1f(x)) = \sigma(\beta)$. Thus $ \sigma(\alpha) \ltimes_{\varPhi(\B)} \sigma(\beta)$. Hence $(f_{\leq 1},\sigma) \in \mathsf{TS_{\ltimes}}(\varPhi(\A),\varPhi(\B))$. We obtain a functor \[\varPhi \colon {\mathsf{HDA_{ext}}}_2 \to \mathsf{TS_{\ltimes}},\; \A \mapsto \varPhi(\A), \; (f,\sigma) \mapsto (f_{\leq 1},\sigma).\]

\begin{prop}
	$\varPhi \dashv \varPsi$. 
\end{prop}

\begin{proof}
	Let $\T$ be a $\ltimes$-transition system. We have 
	\[U(\varPhi(\varPsi(\T))) = \varPsi(\T)_{\leq 1} = U(\T).\] 
	The identity of $U(\T)$ is a morphism of $\ltimes$-transition systems $\varPhi(\varPsi(\T)) \to \T$. Indeed, if $\alpha, \beta \in \Sigma_\T$ such that $\alpha \ltimes_{\varPhi(\varPsi(\T))} \beta$, then there exists an independence square $x = (x^k_i)_{k,i}$ in $\T$ such that $\lambda_\T(x^0_2) = \lambda_{\varPsi(\T)}(d^0_2x) = \alpha$ and $\lambda_{\T}(x^0_1) = \lambda_{\varPsi(\T)}(d^0_1x) = \beta$. Hence $\alpha \ltimes_\T \beta$. We define the counit of the adjunction to be the natural morphism of $\ltimes\mbox{-}$transition systems $\varepsilon_\T = id_{U(\T)} \colon \varPhi(\varPsi(\T)) \to \T$.

	Consider a 2-truncated extensional HDA $\A$. We have $\varPsi (\varPhi(\A))_{\leq 1} = U(\varPhi(\A)) = \A_{\leq 1}$. We define the unit $\eta_\A \colon \A \to \varPsi (\varPhi(\A))$ to be the natural morphism of $2\mbox{-}$truncated extensional HDAs $(\theta_\A, id_{\Sigma_\A})$ where $\theta_\A\colon P_\A \to P_{\varPsi (\varPhi(\A))}$ is given by 
	\[\theta_\A(x) = \left\{ \begin{array}{ll} x, & x \in (P_\A)_{\leq 1},\\
	(d^k_ix)_{k,i}, & x \in (P_\A)_2.
	\end{array}\right.\]

	\begin{sloppypar}	
	Let $\A$ be a 2-truncated extensional HDA. Since \[\varPhi(\eta_\A) = \varPhi(\theta_\A,id_{\Sigma_\A}) = ((\theta_\A)_{\leq 1},id_{\Sigma_\A}) = (id_{(P_\A)_{\leq 1}},id_{\Sigma_\A}) = id_{U(\varPhi(\A))},\] we have $\varepsilon_{\varPhi(\A)} \circ \varPhi(\eta_\A) = id_{\varPhi(\A)}$. Let $\T$ be a $\ltimes$-transition system. We have \[\varPsi(\varepsilon_\T) \circ \eta_{\varPsi(\T)} = (\psi(\varepsilon_\T),id_{\Sigma_\T}) \circ (\theta_{\varPsi(\T)},id_{\Sigma_{\varPsi(\T)}}) = (\psi(\varepsilon_\T)\circ \theta_{\varPsi(\T)}, id_{\Sigma_{\varPsi(\T)}}).\]
	Let $x = (x^k_i)_{k,i}$ be an independence square in $\T$. We have $\psi(\varepsilon_\T)\circ \theta_{\varPsi(\T)} (x) = \psi(\varepsilon_\T)((d^k_ix)_{k,i}) = \psi (id_{P_\T}, id_{\Sigma_\T})((x^k_i)_{k,i}) = (x^k_i)_{k,i} = x$. Since we also have ${\psi(id_{P_\T}, id_{\Sigma_\T}) \circ \theta_{\varPsi(\T)} (x)} = x$ for $x \in (P_{\varPsi(\T})_{\leq 1}$, it follows that $\varPsi(\varepsilon_\T) \circ \eta_{\varPsi(\T)} = id_{\varPsi(\T)}$. \qedhere 
	\end{sloppypar}
\end{proof}

\section{HDA models of \texorpdfstring{$\ltimes$}{}-transition systems} \label{SecHDAmod}

In this section we introduce HDA models of $\ltimes$-transition systems. We show that the HDA model of an $\ltimes$-transition system exists functorially and is unique up to isomorphism. We also establish that in the HDA model of a deterministic $\ltimes\mbox{-}$transition system with an asymmetric relation the independence of $n$ transitions is modeled by a unique $n$-cube.

\subsection{HDA models} \label{HDAmodels}

We say that an HDA $\A$ is an \emph{HDA model} of a $\ltimes$-transition system $\T$ if the following conditions hold:
	\begin{description}
		\item[HM1] $\A_{\leq 1} = U(\T)$.
		\item[HM2] For all $x \in (P_\A)_2$, $\lambda_\A(d^0_2x) \ltimes_\T \lambda_\A(d^0_1x)$.		
		\item[HM3] For all $m\geq 2$ and $x,y \in (P_\A)_m$, if $d^k_rx = d^k_ry$ for all $r \in\{1,\dots ,m\}$ and $k \in \{0,1\}$, then $x = y$.		
		\item[HM4] $\A$ is maximal with respect to the properties  HM1-HM3, i.e., $\A$ is not a proper subautomaton of any HDA satisfying HM1-HM3.
	\end{description}
Condition HM1 says that $\A$ is built on top of $\T$ by filling in empty cubes. By condition HM2, an empty square may only be filled in if it is an independence square. Condition HM3 ensures that no empty cube is filled in twice in the same way. By condition HM4, all admissible empty cubes are filled in. 

Condition HM3 has the following equivalent form, which says that a cube is determined by the family of its vertices and edges (like in a 1-coskeleton):

\begin{prop} \label{HM3equiv} 
	An HDA $\A$ satisfies condition HM3 if and only if for all $m\in \N$, the map \[(P_\A)_m \to \mathsf{PCS}(\llbracket 0,1\rrbracket^{\otimes m}_{\leq 1},(P_\A)_{\leq 1}), \quad x \mapsto (x_{\sharp})_{\leq 1}\] is injective.
\end{prop}

\begin{proof}
	Suppose first that $\A$ satisfies HM3. We show by induction that the map of the statement is injective for all $m \in \N$. If $m \leq 1$, it is actually bijective. Suppose that $m\geq 2$ and that the assertion holds for $m-1$. Let $x,y \in (P_\A)_m$ be elements such that $(x_{\sharp})_{\leq 1}  = (y_{\sharp})_{\leq 1}$. Let $r \in\{1,\dots ,m\}$ and $k \in \{0,1\}$. We have $(d^k_rx)_{\sharp}(\iota_{m-1}) = d^k_rx = d^k_rx_{\sharp}(\iota_m) = x_{\sharp}(d^k_r\iota_m) = x_{\sharp}((d^k_r\iota_m)_{\sharp}(\iota_{m-1})) = {x_{\sharp} \circ (d^k_r\iota_m)_{\sharp}(\iota_{m-1})}$ and therefore $(d^k_rx)_{\sharp} = x_{\sharp} \circ (d^k_r\iota_m)_{\sharp}$. For the same reason, $(d^k_ry)_{\sharp} = {y_{\sharp} \circ (d^k_r\iota_m)_{\sharp}}$. For $w \in \llbracket 0,1 \rrbracket ^{\otimes {m-1}}_{\leq 1}$, $(d^k_r\iota_m)_{\sharp}(w) \in \llbracket 0,1 \rrbracket ^{\otimes {m}}_{\leq 1}$ and hence
	\[(d^k_rx)_{\sharp}(w) = x_{\sharp}((d^k_r\iota_m)_{\sharp}(w))= y_{\sharp} ((d^k_r\iota_m)_{\sharp}(w)) = (d^k_ry)_{\sharp}(w).\]
	Thus $((d^k_rx)_{\sharp})_{\leq 1} = ((d^k_ry)_{\sharp})_{\leq 1}$. By the inductive hypothesis, it follows that $d^k_rx = d^k_ry$. Hence, by HM3, $x = y$.
	
	Suppose now that the map of the statement is injective for all $m \in \N$. Let $m\geq 2$, and consider elements $x,y \in (P_\A)_m$ such that $d^k_rx = d^k_ry$ for all $r \in\{1,\dots ,m\}$ and $k \in \{0,1\}$. Then all iterated boundaries of $x$ and $y$ coincide. Therefore $x_{\sharp}$ and $y_{\sharp}$ agree on all iterated boundaries of $\iota_m$. In particular, $(x_{\sharp})_{\leq 1} = (y_{\sharp})_{\leq 1}$. Hence $x = y$.
\end{proof}

\subsection{The functor \texorpdfstring{$\mathscr{A}$}{}}
Recall the adjunctions $T_2 \dashv \varGamma_2$ and $\varPhi \dashv \varPsi$ from Sections \ref{truncation} and \ref{PsiPhi}. The adjunction $T_2 \dashv \varGamma_2$ induces an adjunction \[T_2\colon \mathsf{HDA_{ext}}\to {\mathsf{HDA_{ext}}}_2 \dashv \varGamma_2 \colon {\mathsf{HDA_{ext}}}_2\to \mathsf{HDA_{ext}}.\]  We define the functor \[\mathscr{A} \colon \mathsf{TS}_{\ltimes} \to \mathsf{HDA_{ext}}\] to be the composite $\varGamma_2\circ \varPsi$. Then $\mathscr{A}$ is right adjoint to the composite $\varPhi\circ T_2$.

\begin{thm}
	Let $\T$ ba a $\ltimes$-transition system. Then $\mathscr{A}(\T)$ is an HDA model of $\T$. 
\end{thm}

\begin{proof} (HM1) We have $\mathscr{A}(\T)_{\leq 1} = \varGamma_2(\varPsi (\T))_{\leq 1} = \varPsi (\T)_{\leq 1} = U(\T)$.

(HM2) Let $x\in (P_{\mathscr{A}(\T)})_2 = (P_{\varPsi (\T)})_2$. Then $x$ is an independence square in $\T$, $x = (x^k_i)_{k,i}$. Hence $\lambda_{\mathscr{A}(\T)}(d^0_2x) = \lambda_{\T}(x^0_2) \ltimes_\T \lambda_{\T}(x^0_1) = \lambda_{\mathscr{A}(\T)}(d^0_1x)$.

(HM3) Let $m\geq 2$ and $x,y \in (P_{\mathscr{A}(\T)})_m$ such that $d^k_rx = d^k_ry$ for all $r \in\{1,\dots ,m\}$ and $k \in \{0,1\}$. If $m=2$, then $x$ and $y$ are independence squares in $\T$, $x = (x^k_i)_{k,i}$, $y = (y^k_i)_{k,i}$, and we have $x^k_i = d^k_ix = d^k_iy = y^k_i$ for all $k$ and $i$, i.e., $x = y$. Suppose that $m > 2$. In order to establish that $x=y$ in $ P_{\mathscr{A}(\T)} = P_{\varGamma_2(\varPsi (\T))} = \varGamma_2(P_{\varPsi (\T)})$, it suffices to show that $(x_{\sharp})_{\leq 2} = (y_{\sharp})_{\leq 2} \in \mathsf{PCS}(\llbracket 0,1 \rrbracket ^{\otimes m}_{\leq 2},P_{\varPsi (\T)})$. This follows from the fact that all iterated boundaries of $x$ and $y$ coincide.

(HM4) Let $\A$ be an HDA satisfying conditions HM1-HM3 such that $\mathscr{A}(\T) \subseteq \A$. Suppose that $\mathscr{A}(\T) \neq \A$. Let $m$ be the minimal degree such that $(P_\A)_m \not \subseteq (P_{\mathscr{A}(\T)})_m$, and consider an element $x\in (P_\A)_m\setminus (P_{\mathscr{A}(\T)})_m$. Then $m \geq 2$ and for all $k$ and $i$, $d^k_ix \in (P_{\mathscr{A}(\T)})_{m-1}$. If $m=2$, then $y = (d^k_ix)_{k,i}$ is an independence square in $\T$ and $y \in (P_{\varPsi(\T)})_2 = (P_{\mathscr{A}(\T)})_2 \subseteq (P_\A)_2$. Since $d^k_iy = d^k_ix$ for all $k$ and $i$, by HM3, $x = y$. Therefore $x \in (P_{\mathscr{A}(\T)})_2$, which is not the case. Hence $m > 2$. Let $z$ be an element of $(P_{\mathscr{A}(\T)})_m =  (\varGamma_2(P_{\varPsi (\T)}))_m$ such that 
\[(z_{\sharp})_{\leq 2} = (x_{\sharp})_{\leq 2} \colon \llbracket 0,1\rrbracket ^{\otimes m}_{\leq 2} \to (P_\A)_{\leq 2} = (P_{\mathscr{A}(\T)})_{\leq 2}  = P_{\varPsi (\T)}.\]
By Proposition \ref{HM3equiv}, we have $z = x$, which is impossible because $z \in  (P_{\mathscr{A}(\T)})_m$. Thus $\mathscr{A}(\T) = \A$.
\end{proof}

\subsection{Uniqueness of HDA models} 

Let $\T$ be a $\ltimes$-transition system. We show that any two HDA models of $\T$ are isomorphic. 

\begin{prop} \label{filling}
	Let $\A$ be an HDA model of $\T$. Consider an integer $m \geq 2$ and $2m$ (not necessarily distinct) elements $x^k_i$ $(k \in \{0,1\}, i \in \{1, \dots ,m\})$ of degree $m-1$ such that $d^k_ix^l_j = d^l_{j-1}x^k_i$ for all $1\leq i<j \leq m$ and $k,l \in\{0,1\}$. If $m=2$, suppose also that $\lambda_\A(x^0_i) = \lambda_\A(x^1_i)$ for all $i \in \{1,2\}$ and that $\lambda_\A(x^0_2) \ltimes_\T \lambda_\A(x^0_1)$. Then there exists a unique element $x$ of degree $m$ such that $d^k_ix = x^k_i$ for all $k \in \{0,1\}$ and $i \in \{1, \dots ,m\}$.
\end{prop}

\begin{proof}
	Uniqueness follows from condition HM3. Suppose that there is no element $x \in (P_\A)_m$ such that $d^k_ix = x^k_i$  for all $k$ and $i$. Consider the HDA $\B$ given by $P_\B = P_\A \cup \{x\}$ where $x \notin P_\A$, $\deg(x) = m$, and $d^k_ix = x^k_i$, $I_\B = I_\A$, $F_\B = F_\A$, $\Sigma_\B = \Sigma_\A$, and $\lambda_\B = \lambda_\A$. Then $\B$ satisfies the conditions HM1-HM3. We obtain a contradiction because $\A$ is both maximal and a proper subautomaton of $\B$.
\end{proof}

\begin{prop} \label{exmor}
	Let $\A$ be an HDA satisfying conditions HM1 and HM2 with respect to a $\ltimes$-transition system $\S$, and let $\B$ be an HDA model of $\T$. Then for every morphism of $\ltimes$-transition systems $\S \to \T$, there exists a unique  morphism of HDAs $\A \to \B$ that restricts to $U(\S) \to U(\T)$ in degrees $\leq 1$. 
\end{prop}

\begin{proof}
	Let $(f,\sigma)\colon \S \to \T$ be a morphism of $\ltimes$-transitions systems. We show inductively that $f$ extends uniquely to degrees $\geq 2$. Let $m \geq 2$, and suppose that $f$ extends uniquely up to degree $m-1$. Let $x \in (P_\A)_m$. Consider the elements $y^k_i = f(d^k_ix)$ $(k \in \{0,1\}, \, i \in \{1, \dots,m\})$. Then for $1 \leq i < j \leq m$ and $k, l \in \{0,1\}$, $d^k_iy^l_j = d^k_if(d^l_jx) = f(d^k_id^l_jx) = f(d^l_{j-1}d^k_ix) = d^l_{j-1}f(d^k_ix) = d^l_{j-1}y^k_i$. If $m=2$, we also have for $i \in \{1,2\}$, $\lambda_\B (y^0_i) = \lambda_\T(f(d^0_ix)) =  \sigma (\lambda_\S(d^0_ix)) =  \sigma (\lambda_\A(d^0_ix)) =  \sigma (\lambda_\A(d^1_ix)) =  \sigma (\lambda_\S(d^1_ix)) = \lambda_\T(f(d^1_ix))  = \lambda_\B (y^1_i)$. Moreover, if $m=2$,  $\lambda_\S(d^0_2x) \ltimes_\S \lambda_\S(d^0_1x)$ and therefore \[\lambda_\B(y^0_2) = \lambda_\T(f(d^0_2x)) = \sigma(\lambda_\S(d^0_2x)) \ltimes_\T \sigma(\lambda_\S(d^0_1x)) = \lambda_\T(f(d^0_1x)) = \lambda_\B(y^0_1).\] By Proposition \ref{filling}, there exists a unique element $y \in (P_\B)_m$ such that $d^k_iy = y^k_i$ for all $k \in \{0,1\}$ and $i \in \{1, \dots ,m\}$. Set $f(x) = y$. We obtain that $f$ extends uniquely to degree $m$. 
\end{proof}

\begin{cor} \label{isomodel}
	Given two HDA models $\A$ and $\B$ of $\T$, there exists a unique morphism of HDAs $\A \to \B$ that restricts to the identity on $U(\T)$, and this morphism is an isomorphism.
\end{cor}

\subsection{Cubes and independence}

Cubes of dimension $\geq 2$ in an HDA represent independence of actions. It follows from Theorem \ref{paradigm} below that in an HDA model of a deterministic $\ltimes$-transition system $\T$ such that $\ltimes_\T$ is asymmetric the independence of $n$ actions at a state is represented by a unique $n$-cube. 

\begin{lem} \label{paralleledges} 
	Let $\A$ be an HDA, and let $x$ be a cube  of degree $m \geq 2$ of $\A$. Consider an edge $d^{k_1}_{1} \cdots d^{k_{j-1}}_{j-1}d^{k_{j+1}}_{j+1}\cdots d^{k_m}_mx$, and let $l_1, \dots, l_{j-1}, l_{j+1} , \dots, l_m \in \{0,1\}$. Then \[\lambda_\A(d^{k_1}_{1} \cdots d^{k_{j-1}}_{j-1}d^{k_{j+1}}_{j+1}\cdots d^{k_m}_mx) = \lambda_\A( d^{l_1}_{1} \cdots d^{l_{j-1}}_{j-1}d^{l_{j+1}}_{j+1}\cdots d^{l_m}_mx).\] 
\end{lem}

\begin{proof}
	It is enough to consider the case where $k_r = l_r$ for all $r$ but one, say $i$. Suppose first that $i < j$. Consider the $2$-cube \[y = d^{k_1}_{1} \cdots d^{k_{i-1}}_{i-1}d^{k_{i+1}}_{i+1} \cdots d^{k_{j-1}}_{j-1}d^{k_{j+1}}_{j+1}\cdots d^{k_m}_mx.\] Then $d^{k_i}_1y = d^{k_1}_{1} \cdots d^{k_{j-1}}_{j-1}d^{k_{j+1}}_{j+1}\cdots d^{k_m}_mx$ and $d^{l_i}_1y = d^{l_1}_{1} \cdots d^{l_{j-1}}_{j-1}d^{l_{j+1}}_{j+1}\cdots d^{l_m}_mx$. Thus \[\lambda_\A(d^{k_1}_{1} \cdots d^{k_{j-1}}_{j-1}d^{k_{j+1}}_{j+1}\cdots d^{k_m}_mx) = \lambda_\A( d^{l_1}_{1} \cdots d^{l_{j-1}}_{j-1}d^{l_{j+1}}_{j+1}\cdots d^{l_m}_mx).\] The case $i > j$ is analogous. 	
\end{proof}

\begin{prop} \label{alphaprop}
	\begin{sloppypar}
	Let $\A$ be an HDA satisfying conditions HM1 and HM2 with respect to a $\ltimes\mbox{-}$transition system $\T$. Let $x$ be a cube of degree $m \geq 2$ of $\A$, and let $d^{k_1}_{1} \cdots d^{k_{i-1}}_{i-1}d^{k_{i+1}}_{i+1}\cdots d^{k_m}_mx$ and $d^{l_1}_{1} \cdots d^{l_{j-1}}_{j-1}d^{l_{j+1}}_{j+1}\cdots d^{l_m}_mx$ be edges of $x$ such that $i < j$. Then 
	\[\lambda_\A(d^{k_1}_{1} \cdots d^{k_{i-1}}_{i-1}d^{k_{i+1}}_{i+1}\cdots d^{k_m}_mx) \ltimes_\T  \lambda_\A(d^{l_1}_{1} \cdots d^{l_{j-1}}_{j-1}d^{l_{j+1}}_{j+1}\cdots d^{l_m}_mx).\] 
	\end{sloppypar}
\end{prop}

\begin{proof}
	By Lemma \ref{paralleledges}, we may suppose that all $k_r$ and $l_r$ are $0$. Consider the $2$-cube \[y = d^{0}_{1} \cdots d^{0}_{i-1}d^{0}_{i+1} \cdots d^{0}_{j-1}d^{0}_{j+1}\cdots d^{0}_mx.\] Then $d^0_1y = d^{0}_{1} \cdots  d^{0}_{j-1}d^{0}_{j+1}\cdots d^{0}_mx$ and $d^0_2y = d^{0}_{1} \cdots  d^{0}_{i-1}d^{0}_{i+1}\cdots d^{0}_mx$. By condition HM2, 
	$\lambda_\A(d^{0}_{1} \cdots  d^{0}_{i-1}d^{0}_{i+1}\cdots d^{0}_mx) \ltimes_\T  \lambda_\A(d^{0}_{1} \cdots  d^{0}_{j-1}d^{0}_{j+1}\cdots d^{0}_mx)$. 
\end{proof}

\begin{lem} \label{paradigmlem}
	Let $\A$ be an  HDA satisfying conditions HM1 and HM2 with respect to a deterministic $\ltimes$-transition system $\T$ such that $\ltimes_\T$ is asymmetric. Let $x$ and $y$ be two $n\mbox{-}$cubes of $\A$ $(n\geq 2)$ such that 
	\[\{d^0_1\cdots d^0_{i-1}d^0_{i+1}\cdots d^0_nx \,|\, i\in \{1,\dots,n\}\} = \{d^0_1\cdots d^0_{i-1}d^0_{i+1}\cdots d^0_ny \,|\, i\in \{1,\dots,n\}\}.\]
	Then $(x_{\sharp})_{\leq 1} = (y_{\sharp})_{\leq 1}$. 
\end{lem}

\begin{proof}
	\begin{sloppypar}
	It suffices to show that for all $i \in \{1, \dots, n\}$ and $k_1, \dots, k_{i-1},k_{i+1}, \dots,k_n \in \{0,1\}$, $d^{k_1}_1 \cdots d^{k_{i-1}}_{i-1}d^{k_{i+1}}_{i+1}\cdots d^{k_n}_nx = d^{k_1}_1 \cdots d^{k_{i-1}}_{i-1}d^{k_{i+1}}_{i+1}\cdots d^{k_n}_ny$. We proceed by induction on the number of indexes $j$ such that $k_j = 1$. Suppose first that all $k_j = 0$. The edges $d^0_1 \cdots d^0_{i-1}d^0_{i+1}\cdots d^0_nx$ are distinct. Indeed, for $i<j$, by Proposition \ref{alphaprop}, \[\lambda_\A(d^0_1 \cdots d^0_{i-1}d^0_{i+1}\cdots d^0_nx) \ltimes_\T \lambda_\A(d^0_1 \cdots d^0_{j-1}d^0_{j+1}\cdots d^0_nx).\] Since $\ltimes_\T$ is irreflexive, it follows that ${d^0_1 \cdots d^0_{i-1}d^0_{i+1}\cdots d^0_nx} \not= {d^0_1 \cdots d^0_{j-1}d^0_{j+1}\cdots d^0_nx}$. The same argument shows that the edges $d^0_1 \cdots d^0_{i-1}d^0_{i+1}\cdots d^0_ny$ are distinct. It follows that there exists a permutation $\phi \in S_n$ such that \[{d^0_1 \cdots d^0_{i-1}d^0_{i+1}\cdots d^0_nx} = {d^0_1 \cdots d^0_{\phi(i)-1}d^0_{\phi(i)+1}\cdots d^0_ny}.\] Let $i< j$. By Proposition \ref{alphaprop}, we have 
	\begin{align*}
	\MoveEqLeft{\lambda_\A(d^0_1 \cdots d^0_{\phi(i)-1}d^0_{\phi(i)+1}\cdots d^0_ny) = \lambda_\A(d^0_1 \cdots d^0_{i-1}d^0_{i+1}\cdots d^0_ny)}\\ &\ltimes_\T \lambda_\A(d^0_1 \cdots d^0_{j-1}d^0_{j+1}\cdots d^0_ny) =  \lambda_\A(d^0_1 \cdots d^0_{\phi(j)-1}d^0_{\phi(j)+1}\cdots d^0_ny).
	\end{align*}
	Since $\ltimes_\T$ is asymmetric, it follows, again by Proposition \ref{alphaprop}, that $\phi(i) < \phi(j)$. This implies that $\phi = id$. 
	\end{sloppypar}
	Let $1 \leq m < n$, and suppose inductively that \[d^{k_1}_1 \cdots d^{k_{i-1}}_{i-1}d^{k_{i+1}}_{i+1}\cdots d^{k_n}_nx = d^{k_1}_1 \cdots d^{k_{i-1}}_{i-1}d^{k_{i+1}}_{i+1}\cdots d^{k_n}_ny\] if less than $m$ of the $k_j$ are equal to $1$. Consider two edges $d^{k_1}_1 \cdots d^{k_{i-1}}_{i-1}d^{k_{i+1}}_{i+1}\cdots d^{k_n}_nx$ and $d^{k_1}_1 \cdots d^{k_{i-1}}_{i-1}d^{k_{i+1}}_{i+1}\cdots d^{k_n}_ny$ where exactly $m$ of the $k_j$ are equal to $1$. Choose an index $j$ such that $k_j = 1$, and set $k_i = 0$. By the inductive hypothesis, $d^{k_1}_1 \cdots d^{k_{j-1}}_{j-1}d^{k_{j+1}}_{j+1}\cdots d^{k_n}_nx = d^{k_1}_1 \cdots d^{k_{j-1}}_{j-1}d^{k_{j+1}}_{j+1}\cdots d^{k_n}_ny$. Hence
	\begin{align*}
		\MoveEqLeft{d^0_1d^{k_1}_1 \cdots d^{k_{i-1}}_{i-1}d^{k_{i+1}}_{i+1}\cdots d^{k_n}_nx
		= d^{k_1}_1 \cdots  d^{k_n}_nx
		= d^1_1d^{k_1}_1 \cdots d^{k_{j-1}}_{j-1}d^{k_{j+1}}_{j+1}\cdots d^{k_n}_nx}\\ &= d^1_1d^{k_1}_1 \cdots d^{k_{j-1}}_{j-1}d^{k_{j+1}}_{j+1}\cdots d^{k_n}_ny
		= d^{k_1}_1 \cdots  d^{k_n}_ny
		= d^0_1d^{k_1}_1 \cdots  d^{k_{i-1}}_{i-1}d^{k_{i+1}}_{i+1}\cdots d^{k_n}_ny.
	\end{align*}
	 By Lemma \ref{paralleledges} and the inductive hypothesis,
	 \begin{align*}
	 \MoveEqLeft{\lambda_\A(d^{k_1}_1 \cdots d^{k_{i-1}}_{i-1}d^{k_{i+1}}_{i+1}\cdots d^{k_n}_nx) = \lambda_\A(d^0_1 \cdots d^0_{i-1}d^0_{i+1}\cdots d^0_nx)}\\
	 &= \lambda_\A(d^0_1 \cdots d^0_{i-1}d^0_{i+1}\cdots d^0_ny) = \lambda_\A(d^{k_1}_1 \cdots d^{k_{i-1}}_{i-1}d^{k_{i+1}}_{i+1}\cdots d^{k_n}_nx).
	 \end{align*}
	 Since $\T$ is deterministic, it follows that \[d^{k_1}_1 \cdots d^{k_{i-1}}_{i-1}d^{k_{i+1}}_{i+1}\cdots d^{k_n}_nx = d^{k_1}_1 \cdots d^{k_{i-1}}_{i-1}d^{k_{i+1}}_{i+1}\cdots d^{k_n}_ny. \qedhere\]
\end{proof}

\begin{thm} \label{paradigm}
	Let $\A$ be an HDA satisfying conditions HM1-HM3 with respect to a deterministic $\ltimes$-transition system $\T$ such that $\ltimes_\T$ is asymmetric. Consider a vertex $v \in (P_\A)_0$ and $n$ action labels $\alpha_1, \dots, \alpha_n \in \Sigma_\A$ $(n \geq 2)$. Then there exists at most one cube $x \in (P_\A)_n$ such that $d^0_1 \cdots d^0_nx = v$ and 
	\[\{\lambda_\A(d^0_1\cdots d^0_{i-1}d^0_{i+1}\cdots d^0_nx) \,|\, i\in \{1,\dots,n\}\} = \{\alpha_1, \dots, \alpha_n\}.\]	
\end{thm}

\begin{proof}
	Consider cubes $x, y \in (P_\A)_n$ having the properties of the statement. Then for each edge $d^0_1\cdots d^0_{i-1}d^0_{i+1}\cdots d^0_nx$, there exists an edge  $d^0_1\cdots d^0_{j-1}d^0_{j+1}\cdots d^0_ny$ with the same label and the same starting point. Since $\T$ is deterministic, the two edges coincide. Hence 
	\[\{d^0_1\cdots d^0_{i-1}d^0_{i+1}\cdots d^0_nx \,|\, i\in \{1,\dots,n\}\} = \{d^0_1\cdots d^0_{i-1}d^0_{i+1}\cdots d^0_ny \,|\, i\in \{1,\dots,n\}\}.\] 
	By Lemma \ref{paradigmlem}, we have $(x_{\sharp})_{\leq 1} = (y_{\sharp})_{\leq 1}$. By Proposition \ref{HM3equiv}, it follows that $x = y$.
\end{proof}

\section{Parallel composition} \label{SecParallel}

The parallel composition of completely independent shared-variable systems can be modeled using the interleaving operator for $\ltimes$-transition systems, which we introduce in this section. We show that an HDA model of the interleaving of two $\ltimes$-transition systems is given by the tensor product of their HDA models.

\subsection{Tensor product of HDAs} The \emph{tensor product} of two HDAs $\A$  and $\B$ is the HDA $\A\otimes \B$ defined by $P_{\A\otimes \B}  = P_{\A}\otimes P_{\B}$, $I_{\A \otimes \B} = (I_{\A},I_{\B})$, $F_{\A \otimes \B} = F_{\A}\times F_{\B}$, $\Sigma_{\A \otimes \B} =  \Sigma_{\A}\amalg \Sigma_{\B}$, and
\[\lambda_{\A\otimes \B} (x,y) = \left \{\begin{array}{ll}
\lambda_{\A}(x), & (x,y) \in (P_{\A})_1\times (P_{\B})_0,\\
\lambda_{\B}(y), & (x,y) \in (P_{\A})_0\times (P_{\B})_1.
\end{array}\right.\]
We remark that the categories $\mathsf{HDA}$ and $\mathsf{HDA_{ext}}$ are monoidal categories with respect to this tensor product.

\subsection{Interleaving of \texorpdfstring{$\ltimes$}{}-transition systems}
\begin{sloppypar}
We define the \emph{interleaving} of two $\ltimes\mbox{-}$transition systems  $\mathcal{S}$ and $\T$ to be the $\ltimes\mbox{-}$transition system $\mathcal{S} \interleave \T$ where \[U(\mathcal{S} \interleave \T) = (U(\S) \otimes U(\T))_{\leq 1}\] and $\alpha \ltimes_{\S \interleave \T} \beta$ if either 
\begin{enumerate}
	\item $\alpha, \beta \in \Sigma_\S$ and $\alpha\ltimes_\S \beta$;
	\item $\alpha, \beta \in \Sigma_\T$ and $\alpha\ltimes_\T \beta$; or
	\item $\alpha \in \Sigma_\S$ and $\beta \in \Sigma_\T$.
\end{enumerate}
With respect to the interleaving operator $\interleave$, the category of $\ltimes$-transition systems is a monoidal category.
\end{sloppypar}

\subsection{Parallel composition of shared-variable systems}
Let $\P^1= (\P^1_1, \dots, \P^1_n)$ and $\P^2 = (\P^2_1, \dots, \P^2_m)$ be two shared-variable systems given by program graphs over the set of variables $Var$. The \emph{parallel composition} of $\P^1$ and $\P^2$ is the shared-variable system $\P = (\P_1, \dots, \P_{n+m})$ given by
\[\P_i = \left\{ \begin{array}{ll}
\vspace{0.2cm}
\P^1_i, & 1 \leq i \leq n,\\
\P^2_{i-n}, & n < i \leq n+m.
\end{array}\right.\]
We shall write $V^1$, $V^2$, and $V$ to denote the sets of declared variables of the systems $\P^1$, $\P^2$, and $\P$, respectively, and use corresponding notations for their state graphs, their initial states, and the ingredients of their program graphs. Let $\eta_0 \in Ev(Var)$ be an initial evaluation. 

\begin{thm} \label{iso1} 
If $V^1 \cap V^2 = \emptyset$, then
\[\T_{\P_1, \dots ,\P_{n+m}} \cong \T_{\P^1_1, \dots ,\P^1_n} \interleave \T_{\P^2_1, \dots ,\P^2_m}.\]
\end{thm}

\begin{proof}
We write $\T = \T_{\P_1, \dots ,\P_{n+m}}$, $\T^1 = \T_{\P^1_1, \dots ,\P^1_n}$, and $\T^2=\T_{\P^2_1, \dots ,\P^2_m}$. Since $V^1\cap V^2 = \emptyset$, we may define a morphism of precubical sets \[\varphi\colon (Q^1\otimes Q^2)_{\leq 1} \to Q\] by
\[((l^1_1, \dots ,l^1_{n},\eta^1),(l^2_1, \dots ,l^2_{m},\eta^2)) \mapsto (l^1_1, \dots ,l^1_{n},l^2_1, \dots ,l^2_{m},\eta^1\ast \eta^2)\]
in degree $0$ and by 
\begin{align*}
\MoveEqLeft{((l^1_1, \dots, l^1_{i-1},t,l^1_{i+1}, \dots  ,l^1_{n},\gamma,\eta^1),(l^2_1, \dots ,l^2_{m},\eta^2))}\\ & \mapsto (l^1_1, \dots, l^1_{i-1},t,l^1_{i+1} ,\dots, l^1_{n},l^2_1, \dots ,l^2_{m},\gamma,\eta^1\ast \eta^2)
\end{align*}
and 
\begin{align*}
\MoveEqLeft{((l^1_1, \dots ,l^1_{n},\eta^1),(l^2_1, \dots, l^2_{i-1},t,l^2_{i+1}, \dots  ,l^2_{m},\gamma,\eta^2))}\\ & \mapsto (l^1_1, \dots ,l^1_{n},l^2_1, \dots, l^2_{i-1},t,l^2_{i+1} ,\dots, l^2_{m},\gamma,\eta^1\ast \eta^2)
\end{align*}
in degree $1$. It is straightforward to check that $\varphi$ is an isomorphism and that $\varphi (I^1, I^2) = I$. Using the relation between $(Q^1\otimes Q^2)^{\mathbb I}$ and $(Q^1)^{\mathbb I}\times (Q^2)^{\mathbb I}$ (see Section \ref{paths}), one easily establishes that the vertices of $P_{\T^1} \otimes P_{\T^2}$ are precisely the elements of $(Q^1\otimes Q^2)_0$ to which there exists a path from $(I^1, I^2)$. This implies that $\varphi$ restricts to an isomorphism of precubical sets 
\[f\colon P_{\T^1 \interleave \T^2} = (P_{\T^1} \otimes P_{\T^2})_{\leq 1} \xrightarrow{\cong} P_\T.\]
We have $f(I_{\T^1 \interleave \T^2}) = f(I_{\T^1},I_{\T^2}) = \varphi (I^1,I^2) = I = I_\T$ and $f(F_{\T^1\interleave \T^2})  = \emptyset = F_\T$. 
Let $\sigma \colon \Sigma_{\T^1\interleave \T^2} = \Sigma_{\T^1} \amalg \Sigma_{\T^2} \to \Sigma_\T$ be the bijection given by $\sigma (i,a^1) = (i,a^1)$ and $\sigma (i,a^2) = (n+i,a^2)$. Then 
$\lambda_\T \circ f(x,y) = \sigma \circ \lambda_{\T^1\interleave \T^2}(x,y)$ for all edges $(x,y) \in {(P_{\T^1} \otimes P_{\T^2})_{1}}$. It follows that $(f,\sigma) \colon {U(\T^1 \interleave \T^2)} \to U(\T)$ is an isomorphism of HDAs. To conclude that $(f,\sigma)$  is an isomorphism of $\ltimes$-transition systems $\T^1 \interleave \T^2 \to \T$, it remains to show that for all $\alpha, \beta \in \Sigma_{\T^1\interleave \T^2}$, 
\[\alpha \ltimes _{\T^1\interleave \T^2} \beta \Leftrightarrow \sigma(\alpha) \ltimes _{\T} \sigma(\beta).\]

Suppose first that $\alpha \ltimes _{\T^1\interleave \T^2} \beta$.  If $\alpha, \beta \in \Sigma_{\T^1}$ and $\alpha \ltimes _{\T^1} \beta$, then there exist integers $1 \leq i < j \leq n$ and actions $a \in A^1_i$ and $b\in A^1_j$ such that $\alpha = (i,a)$ and $\beta = (j,b)$. Since $\sigma(\alpha) = (i,a)$, $\sigma(\beta) = (j,b)$, and $i <j$, we have $\sigma(\alpha) \ltimes _{\T} \sigma(\beta)$. A similar argument shows that $\sigma(\alpha) \ltimes _{\T} \sigma(\beta)$ if $\alpha, \beta \in \Sigma_{\T^2}$ and $\alpha \ltimes _{\T^2} \beta$. If $\alpha \in \Sigma_{\T^1}$ and $\beta \in \Sigma_{\T^2}$, then there exist integers $1 \leq i \leq n$ and $1 \leq j \leq m$ and actions $a \in A^1_i$ and $b\in A^2_j$ such that $\alpha = (i,a)$ and $\beta = (j,b)$. Since $\sigma(\alpha) = (i,a)$, $\sigma(\beta) = (n+j,b)$, and $i < n+j$, we have $\sigma(\alpha) \ltimes _{\T} \sigma(\beta)$.

Suppose now that $\sigma(\alpha) \ltimes_\T \sigma(\beta)$. If $\alpha, \beta \in \Sigma_{\T^1}$, then there exist integers $1 \leq i, j \leq n$ and actions $a \in A^1_i$ and $b\in A^1_j$ such that $\alpha = (i,a)$ and $\beta = (j,b)$. Since $\sigma(\alpha) = (i,a)$, $\sigma(\beta) = (j,b)$, we have $i <j$ and therefore $\alpha \ltimes_{\T^1} \beta$. Thus $\alpha \ltimes _{\T^1\interleave \T^2} \beta$. A similar argument shows that $\alpha \ltimes _{\T^1\interleave \T^2} \beta$ if $\alpha, \beta \in \Sigma_{\T^2}$. If $\alpha \in \Sigma_{\T^1}$ and $\beta \in \Sigma_{\T^2}$, then $\alpha \ltimes _{\T^1\interleave \T^2} \beta$ in any case. It is not possible that $\alpha \in \Sigma_{\T^2}$ and $\beta \in \Sigma_{\T^1}$. Indeed, this would imply $\beta \ltimes _{\T^1\interleave \T^2} \alpha$ and, by what we have already seen, $\sigma(\beta) \ltimes_\T \sigma(\alpha)$. Since $\ltimes_\T$ is asymmetric, this is impossible. 
\end{proof}

\begin{rem}
Although this will, of course, not happen in general, one may have an isomorphism 
$\T_{\P_1, \dots ,\P_{n+m}} \cong \T_{\P^1_1, \dots ,\P^1_n} \interleave \T_{\P^2_1, \dots ,\P^2_m}$ 
even if $V^1 \cap V^2 \not= \emptyset$.	For example, for the system $(\P_1,\P_2)$ of Example \ref{specex}, one has $\T_{\P_1,\P_2} \cong \T_{\P_1} \interleave \T_{\P_2}$ although the variable $x$ is used by both $\P_1$ and $\P_2$.

\end{rem}

\subsection{HDA models of interleavings} \label{secHDAinter} Let $\mathcal S$ and $\T$ be $\ltimes$-transition systems, and let $\A$ and $\B$ be HDA models of $\mathcal S$ and $\T$, respectively. We establish that the tensor product $\A \otimes \B$ is an HDA model of the interleaving $\mathcal{S} \interleave \T$. Given its intuitive plausibility, it is perhaps surprising that this result does not hold in either the nonextensional or the precubical setting. Consider, for example, a 1-truncated HDA with two vertices and two identically labeled edges from one to the other. While the tensor product of this  nonextensional HDA with an HDA representing an interval has two 2-cubes, its 1-skeleton has four independence squares. Hence the HDA corresponding to the HDA model of the interleaving would have four $2\mbox{-}$cubes. This shows that the result that  $\A \otimes \B$ is an HDA model of $\mathcal{S} \interleave \T$ does not hold without extensionality. It does not hold in the precubical setting either, as is shown by the same example without labels or the even simpler one of two intervals. As a consequence, this result cannot be established by a general abstract argument that also would apply to precubical sets or nonextensional HDAs. 

\begin{prop} \label{conditions} 
	The HDA $\A \otimes \B$ satisfies conditions HM1-HM3 for the $\ltimes\mbox{-}$transition system $\mathcal{S} \interleave \T$.
\end{prop}

\begin{proof} 
	(HM1) We have  $(\A\otimes \B) _{\leq 1} = (\A_{\leq 1}\otimes \B_{\leq 1}) _{\leq 1} =  (U(\S) \otimes U(\T))_{\leq 1} = U(\S \interleave \T)$.
	
	(HM2)  Let $(x,y) \in P_{\A \otimes \B}$ be an element of degree $2$.
	If $(x,y) \in (P_\A)_2\times (P_\B)_0$, then, by condition HM2 for $\A$, $\lambda_{\A}(d^0_2x) \ltimes_\S \lambda_{\A}(d^0_1x)$. Hence 
	\begin{align*}
	\MoveEqLeft{\lambda_{\A\otimes \B}(d^0_2(x,y)) = \lambda_{\A\otimes \B}(d^0_2x,y) = \lambda_{\A}(d^0_2x)}\\ &\ltimes_{\S \interleave \T} \lambda_{\A}(d^0_1x) = \lambda_{\A\otimes \B}(d^0_1x,y) = \lambda_{\A\otimes \B}(d^0_1(x,y)).
	\end{align*}
	A similar argument shows that $\lambda_{\A\otimes \B}(d^0_2(x,y)) \ltimes_{\S \interleave \T} \lambda_{\A\otimes \B}(d^1_2(x,y))$ if $(x,y) \in (P_\A)_0\times (P_\B)_2$. If $(x,y) \in (P_\A)_1\times (P_\B)_1$, then 
	\begin{align*}
	\MoveEqLeft{\lambda_{\A\otimes \B}(d^0_2(x,y)) = \lambda_{\A\otimes \B}(x,d^0_1y) = \lambda_{\A}(x)}\\ &\ltimes_{\S \interleave \T} \lambda_{\B}(y) = \lambda_{\A\otimes \B}(d^0_1x,y) = \lambda_{\A\otimes \B}(d^0_1(x,y)). 
	\end{align*}

	(HM3) Let $m \geq 2$ be an integer. Consider elements $(x,y), (a,b) \in (P_{\A \otimes \B})_m$ such that ${d^k_i(x,y) = d^k_i(a,b)}$ for all  $k \in \{0,1\}$ and $i \in \{1, \dots ,m\}$. Let $s, t \in \{0, \dots, m\}$ be the unique integers such that $(x,y) \in (P_\A)_s\times (P_\B)_{m-s}$ and $(a,b) \in {(P_\A)_t \times (P_\B)_{m-t}}$. We have $s=t$. Indeed, suppose that this not the case. If $s = 0$, then $t > 0$. Since $(x,d^0_1y) = d^0_1(x,y) = d^0_1(a,b) = (d^0_1a,b)$, we have $s = t-1$ and therefore $t=1$. Hence $(x,d^0_my) = d^0_m(x,y) = d^0_m(a,b) = (a,d^0_{m-1}b)$. This implies that $s=t$, which we supposed not to be the case. Therefore $s > 0$. For the same reason, $t > 0$. It follows that $(d^0_1x,y) = d^0_1(x,y) = d^0_1(a,b) = (d^0_1a,b)$, which implies that $s=t$. This is a contradiction. Thus $s=t$. 
								
	If $s = t = 0$, then for all $k \in \{0,1\}$ and $i \in \{1, \dots ,m\}$, $(x,d^k_iy) = d^k_i(x,y) = d^k_i(a,b) = (a,d^k_ib)$. Hence $x=a$ and, by condition HM3 for $\B$, $y=b$. An analogous argument shows that $(x,y) = (a,b)$ if $s = t = m$. It remains to consider the case $0 < s=t < m$. In this case we have $(x,d^0_{m-s}y) = d^0_m(x,y) = d^0_m(a,b) = (a,d^0_{m-t}b)$
	and $(d^0_1x,y) = d^0_1(x,y) = d^0_1(a,b) = (d^0_1a,b)$ and therefore $x=a$ and $y=b$.	
\end{proof}

It remains to show that $\A\otimes \B$ satisfies HM4. For the remainder of this section, let $\C$ be an HDA that satisfies conditions HM1-HM3 for $\S \interleave \T$ and that contains $\A \otimes \B$ as a subautomaton. We have to show that $\A\otimes \B = \C$.

\begin{lem} \label{surjective0}
	Let $r\geq 2$, and let $c \in (P_\C)_r$.
	\begin{enumerate}[(i)]
		\item If $d^k_ic \in (P_\A)_0 \times (P_\B)_{r-1}$ for all $k \in \{0,1\}$ and $i \in \{1, \dots, r\}$, then $c \in (P_\A)_0\times (P_\B)_r$. 
		\item If $d^k_ic \in (P_\A)_{r-1} \times (P_\B)_0$ for all $k \in \{0,1\}$ and $i \in \{1, \dots, r\}$, then $c \in (P_\A)_r\times (P_\B)_0$. 
	\end{enumerate}
\end{lem}

\begin{proof} 
	We prove statement (i) and leave the analogous proof of statement (ii) to the reader. Suppose that $d^k_ic \in (P_\A)_0 \times (P_\B)_{r-1}$ for all $k \in \{0,1\}$ and $i \in \{1, \dots, r\}$. Write $d^k_ic = (x^k_i,y^k_i)$. Then for $1\leq i<j\leq r$ and $k,l \in \{0,1\}$, 
	\[(x^k_i, d^l_{j-1}y^k_i) = d^l_{j-1}(x^k_i,y^k_i) = d^l_{j-1}d^k_ic = d^k_id^l_{j}c = d^k_{i}(x^l_j,y^l_j) = (x^l_j,d^k_{i}y^l_j)\] 
	and hence $x^k_i = x^l_j$ and $d^k_{i}y^l_j = d^l_{j-1}y^k_i$. Set $x = x^0_1$. Then $x = x^k_i$ for all $k$ and $i$. If $r = 2$, we have
	\[\lambda_\B(y^k_i) = \lambda_{\A\otimes \B}(x^k_i,y^k_i)= \lambda_C(x^k_i,y^k_i) = \lambda_\C(d^k_ic) \]
	and therefore $\lambda_\B(y^0_i) = \lambda_\B(y^1_i)$. Moreover, $\lambda_\B(y^0_2) \ltimes_{\S\interleave \T} \lambda_\B(y^0_1)$ and therefore $\lambda_\B(y^0_2) \ltimes_{ \T} \lambda_\B(y^0_1)$. By Proposition \ref{filling}, there exists an element $y \in (P_\B)_r$ such that $d^k_iy = y^k_i$. Since \[d^k_i(x,y) = (x,d^k_iy) = (x^k_i,y^k_i) = d^k_ic,\] by condition HM3, we have $c = (x,y) \in {(P_\A)_0\times (P_\B)_r}$.  
\end{proof}

\begin{lem} \label{surjetive2} 
	$(\A \otimes \B)_{\leq 2} = \C_{\leq 2}$.
\end{lem}

\begin{proof}
	Let $c\in (P_\C)_2$. We show that $c\in (P_\A \otimes P_\B)_2$. We have $d^k_ic \in (P_\A\otimes P_\B)_1$ for all $ k \in \{0,1\}$ and $i \in \{1, 2\}$. Write $d^k_ic = (x^k_i,y^k_i)$.  By Lemma \ref{surjective0}, we may suppose that there exist $s,t,a,b$ such that $(x^a_s,y^a_s) \in (P_\A)_1 \times (P_\B)_0 = (P_\S)_1 \times (P_\T)_0$ and $ (x^b_t,y^b_t) \in (P_\A)_0 \times (P_\B)_1 = (P_\S)_0 \times (P_\T)_1$. We have 
	\begin{align*}
	\lambda_{\S}(x^a_s)
	&= \lambda_{\S\interleave \T}(x^a_s,y^a_s)\\
	&= \lambda_{\C}(x^a_s,y^a_s)\\
	&= \lambda_{\C}(d^a_sc)\\
	&= \lambda_{\C}(d^{1-a}_sc)\\
	&= \lambda_{\C}(x^{1-a}_s,y^{1-a}_s)\\
	&= \lambda_{\S\interleave \T}(x^{1-a}_s,y^{1-a}_s)\\
	&= \left\{\begin{array}{ll}
	\vspace{0.2cm} 
	\lambda_{\S}(x^{1-a}_s), & (x^{1-a}_s,y^{1-a}_s) \in (P_\S)_1 \times (P_\T)_0,\\
	\lambda_{\T}(y^{1-a}_s), & (x^{1-a}_s,y^{1-a}_s) \in (P_\S)_0 \times (P_\T)_1.
	\end{array} \right.
	\end{align*}
	Since $\Sigma_\S$ and $\Sigma_\T$ are disjoint in $\Sigma_{\S\interleave \T} = \Sigma_\S \amalg \Sigma_\T$, it follows that $(x^{1-a}_s,y^{1-a}_s) \in (P_\S)_1 \times (P_\T)_0$ and that $\lambda_{\S}(x^a_s) = \lambda_{\S}(x^{1-a}_s)$. Similarly, $(x^{1-b}_t,y^{1-b}_t) \in (P_\S)_0 \times (P_\T)_1$ and $\lambda_{\T}(y^b_t) = \lambda_{\T}(y^{1-b}_t)$. Thus we have $s \not= t$, $\lambda_{\C}(x^0_s,y^0_s) = \lambda_{\S}(x^0_s) \in  \Sigma_\S$, and $\lambda_\C(x^0_{t},y^0_{t}) = \lambda_{\T}(y^0_{t}) \in  \Sigma_\T$. Since 
	\[\lambda_\C(x^0_2,y^0_2)= \lambda_\C (d^0_2c) \ltimes_{\S \interleave \T} \lambda_\C(d^0_1c) = \lambda_\C(x^0_1,y^0_1),\]
	it follows that $s = 2$ and $t = 1$. 
	
	For all $k,l \in \{0,1\}$, we have 
	\[(d^k_1x^l_2,y^l_2) = d^k_1(x^l_2,y^l_2) = d^k_1d^l_2c = d^l_1d^k_1c = d^l_1(x^k_1,y^k_1) = (x^k_1,d^l_1y^k_1).\]
	Thus $d^k_1x^l_2 = x^k_1$ and $d^l_1y^k_1 = y^l_2$. Hence $x^0_2$ and $x^1_2$ have the same endpoints and the same label. By extensionality, this implies that $x^0_2 = x^1_2$. Similarly, $y^0_1 = y^1_1$. Set $x = x^0_2 = x^1_2$ and $y = y^0_1 = y^1_1$. We have
	$d^k_1(x,y) = (d^k_1x,y) =  (x^k_1,y^k_1) = d^k_1c$
	and 
	$d^k_2(x,y)  = (x,d^k_1y) = (x^k_2,y^k_2) = d^k_2c$. 
	By condition HM3, $c = (x,y) \in {(P_\A \otimes P_\B)_2}$.	
\end{proof}

\begin{lem} \label{pidlem} 
	Let $r>2$ and $c \in (P_\C)_r$. Let $k \in \{0,1\}$, $i \in\{1, \dots, r\}$, and $j \in \{0, \dots, r-1\}$ such that $d^k_ic \in (P_\A)_j \times (P_\B)_{r-1-j}$.
	\begin{enumerate}[(i)]
		\item If $j \geq 1$, then $\left \{ \begin{array}{ll}
		\vspace{0.2cm}
		\lambda_\C(d^0_1 \cdots d^0_{j-1}d^0_{j+1} \cdots  d^0_{r}c) \in \Sigma_\A, & i > j,\\
		\lambda_\C(d^0_1  \cdots d^0_{j}d^0_{j+2}  \cdots d^0_{r}c) \in \Sigma_\A, & i \leq j.
		\end{array}\right.$
		\vspace{0.2cm}
		\item If $j < r-1$, then $\left \{ \begin{array}{ll}
		\vspace{0.2cm}
		\lambda_\C(d^0_1 \cdots d^0_{j}d^0_{j+2} \cdots  d^0_{r}c) \in \Sigma_\B, & i > j+1,\\
		\lambda_\C(d^0_1  \cdots d^0_{j+1}d^0_{j+3}  \cdots d^0_{r}c) \in \Sigma_\B, & i \leq j+1.
		\end{array}\right.$
	\end{enumerate}
\end{lem}

\begin{proof}\leavevmode
  \begin{enumerate}[(i)]
	\item Suppose that $j \geq 1$. Write $d^k_ic = (a,b)$. We have 
	\begin{align*}
	\lambda_\C(d^0_1 \cdots d^0_{j-1}d^0_{j+1} \cdots d^0_{r-1}d^k_ic)  
	&= \lambda_\C (d^0_1 \cdots d^0_{j-1}d^0_{j+1} \cdots d^0_{r-1}(a,b))\\
	&= \lambda_\C (d^0_1 \cdots d^0_{j-1}a,d^0_{1} \cdots d^0_{r-1-j}b)\\
	&= \lambda_{\A\otimes \B} (d^0_1 \cdots d^0_{j-1}a,d^0_{1} \cdots d^0_{r-1-j}b)\\
	&= \lambda_{\A} (d^0_1 \cdots d^0_{j-1}a).
	\end{align*}
	On the other hand, by Lemma \ref{paralleledges},
	\begin{align*}
	\MoveEqLeft{\lambda_\C(d^0_1 \cdots d^0_{j-1}d^0_{j+1} \cdots d^0_{r-1}d^k_ic) }\\
	&= \left \{ \begin{array}{ll}
	\vspace{0.2cm}
	\lambda_\C(d^0_1 \cdots d^0_{j-1}d^0_{j+1} \cdots d^0_{i-1}d^k_id^0_{i+1} \cdots d^0_{r}c), & i > j,\\
	\lambda_\C(d^0_1 \cdots d^0_{i-1}d^k_id^0_{i+1} \cdots d^0_{j}d^0_{j+2}  \cdots d^0_{r}c), & i \leq j\\
	\end{array}\right.\\
	&= \left \{ \begin{array}{ll}
	\vspace{0.2cm}
	\lambda_\C(d^0_1 \cdots d^0_{j-1}d^0_{j+1} \cdots  d^0_{r}c), & i > j,\\
	\lambda_\C(d^0_1  \cdots d^0_{j}d^0_{j+2}  \cdots d^0_{r}c), & i \leq j.
	\end{array}\right.
	\end{align*}
	Hence
	\[\left \{ \begin{array}{ll}
	\vspace{0.2cm}
	\lambda_\C(d^0_1 \cdots d^0_{j-1}d^0_{j+1} \cdots  d^0_{r}c) = \lambda_{\A} (d^0_1 \cdots d^0_{j-1}a) \in \Sigma_\A, & i > j,\\
	\lambda_\C(d^0_1  \cdots d^0_{j}d^0_{j+2}  \cdots d^0_{r}c) = \lambda_{\A} (d^0_1 \cdots d^0_{j-1}a) \in \Sigma_\A, & i \leq j. 
	\end{array} \right. 
	\]
	\item Analogous. \qedhere
  \end{enumerate}
\end{proof}

\begin{lem} \label{alphalem}
	Let $c$ be an element of $(P_\C)_r$ $(r\geq 2)$, and let $p\in \{1,\dots, r\}$ be an integer such that $\lambda_\C(d^0_1 \cdots d^0_{p-1}d^0_{p+1} \cdots  d^0_{r}c) \in \Sigma_\S$. Then $\lambda_\C(d^0_1 \cdots d^0_{j-1}d^0_{j+1} \cdots  d^0_{r}c) \in \Sigma_\S$ for all $j\in \{1, \dots , p\}$.
\end{lem}

\begin{proof}
	Let $j\in \{1, \dots , p-1\}$. By Proposition \ref{alphaprop},
	\[\lambda_\C(d^0_1 \cdots d^0_{j-1}d^0_{j+1} \cdots  d^0_{r}c) \ltimes_{\S \interleave \T} \lambda_\C(d^0_1 \cdots d^0_{p-1}d^0_{p+1} \cdots  d^0_{r}c).\]
	By definition of the relation $\ltimes_{\S \interleave \T}$, it follows that 
	\[\lambda_\C(d^0_1 \cdots d^0_{j-1}d^0_{j+1} \cdots  d^0_{r}c) \in \Sigma_\S.\qedhere\]
\end{proof}

\begin{prop}
	$\A \otimes \B = C$.
\end{prop}

\begin{proof}
	By Lemma \ref{surjetive2}, $(\A \otimes \B)_{\leq 2} = \C_{\leq 2}$. Let $r > 2$, and suppose inductively that ${(\A \otimes \B)_{\leq r-1}} = \C_{\leq r-1}$. Let $c \in (P_\C)_r$.  We show that $c \in (P_\A\otimes P_\B)_r$. By the induction hypothesis, for all $k \in \{0,1\}$ and $i \in \{1, \dots, r\}$, there exist elements $(x^k_i,y^k_i) \in (P_\A \otimes P_\B)_{r-1}$ such that $d^k_ic = (x^k_i,y^k_i)$. By Lemma \ref{surjective0}, we may suppose that for some $k,i$, $d^k_ic \notin {(P_\A)_0 \times (P_\B)_{r-1}}$ and that for some $k',i'$, $d^{k'}_{i'}c \notin {(P_\A)_{r-1} \times (P_\B)_{0}}$. By Lemma \ref{pidlem}(i), there exists an integer $j \in \{1, \dots, r\}$ such that $\lambda_\C(d^0_1 \cdots d^0_{j-1}d^0_{j+1} \cdots d^0_rc) \in \Sigma_\A$. Let $p$ be the largest such $j$. Then $\lambda_\C(d^0_1 \cdots d^0_{j-1}d^0_{j+1} \cdots  d^0_{r}c) \in \Sigma_\B$ for all $j > p$ and, by Lemma \ref{alphalem},  $\lambda_\C(d^0_1 \cdots d^0_{j-1}d^0_{j+1} \cdots  d^0_{r}c) \in \Sigma_\A$ for all $j\leq p$. By Lemma \ref{pidlem}(ii), we have $p \in \{1, \dots, r-1\}$.

	Consider $i > p$ and $k \in \{0,1\}$. We show that $d^k_ic \in (P_\A)_p\times (P_\B)_{r-1-p}$. Suppose that $d^k_ic \in (P_\A)_s\times (P_\B)_{r-1-s}$. It follows from Lemma \ref{pidlem}(i) that $s \leq p$. If we had $s < p$, we would have $s < r-1$ and $i > s+1$. By Lemma \ref{pidlem}(ii), we would obtain $p < s+1 \leq p$, which is impossible. Thus $s = p$.
	
	Consider $i \leq p$ and $k \in \{0,1\}$. We show that $d^k_ic \in (P_\A)_{p-1}\times (P_\B)_{r-p}$. Suppose that $d^k_ic \in (P_\A)_s\times (P_\B)_{r-1-s}$. Since $d^0_rc \in (P_\A)_p\times (P_\B)_{r-1-p}$, we have $d^k_id^0_rc = d^k_i(x^0_r,y^0_r) =  (d^k_ix^0_r,y^0_r)$
	and
	\[
	d^0_{r-1}d^k_ic = d^0_{r-1}(x^k_i,y^k_i)
	= \left \{ \begin{array}{ll}
	\vspace{0.2cm}
	(d^0_{r-1}x^k_i,y^k_i) , & s= r-1,\\
	(x^k_i,d^0_{r-1-s}y^k_i) , & s <r-1.
	\end{array} \right.
	\]
	We obtain that $p-1\leq s \leq p$. Suppose that $s = p$. Then $i \leq s = r-1$. By Lemma \ref{pidlem}(i), this implies that $\lambda_\C(d^0_1 \cdots d^0_{r-1}c) \in \Sigma_\A$. This is impossible because $p < r$. Thus $s = p-1$.
	
	For $1\leq i \leq p < j \leq r$ and $k,l \in \{0,1\}$, $d^k_id^l_jc = d^k_i(x^l_j,y^l_j) =  (d^k_ix^l_j,y^l_j)$
	and $d^l_{j-1}d^k_ic = d^l_{j-1}(x^k_i,y^k_i) = (x^k_i,d^l_{j-p}y^k_i)$.
	We obtain $d^k_ix^l_j = x^k_i$ and $d^l_{j-p}y^k_i = y^l_j$. If $1<p<r-1$, it follows, by condition HM3 for $\A$ and $\B$, that all $x^l_j$ with $j >p$ and all $y^k_i$ with $i \leq p$ are equal. By extensionality, this also holds if $p=1$ or $p=r-1$. Indeed, if $p=1$,  the $x^l_j$ with $j >p$ are not just edges with the same endpoints. By Lemma \ref{paralleledges}, they also have the same label:
	\begin{align*}
	\lambda_\A(x^l_j) &= \lambda_\C(x^l_j, d^0_1\cdots d^0_{r-2}y^l_j)\\
	&= \lambda_\C(d^0_2\cdots d^0_{r-1}(x^l_j, y^l_j))\\
	&= \lambda_\C(d^0_2\cdots d^0_{r-1}d^l_jc)\\
	&= \lambda_\C(d^0_2\cdots d^0_{r}c).
	\end{align*}
	Similarly, if $p = r-1$, all $y^k_i$ with $i \leq p$ have the same label. Set $x = x^0_r$ and $y = y^0_1$. Then for all $k,i$, 
	\[
	d^k_i(x,y) = \left \{ \begin{array}{ll}
	\vspace{0.2cm}
	(d^k_ix,y), & i \leq p,\\
	(x,d^k_{i-p}y), & i > p
	\end{array}\right.
	= (x^k_i,y^k_i) = d^k_ic.
	\]
	Since $\C$ satisfies condition HM3, it follows that $c = (x,y) \in (P_\A)_p\times (P_\B)_{r-p}$. 
\end{proof}

As a consequence we obtain our main result:

\begin{thm} \label{mainresult}
	The tensor product $\A \otimes \B$ is an HDA model of the interleaving $\mathcal{S} \interleave \T$.
\end{thm}

By Theorem \ref{mainresult} and Corollary \ref{isomodel}, there exists a unique isomorphism \[\mathscr{A}(\S) \otimes \mathscr{A}(\T) \xrightarrow{\cong} \mathscr{A}(\S \interleave \T)\] that restricts to the identity on $U(\S \interleave \T)$. Considering this isomorphism, which is natural by Proposition \ref{exmor}, one easily establishes the following result:

\begin{thm}
	The functor $\mathscr{A}$ is a strong monoidal functor.
\end{thm}

\begin{exa}
Consider the system $(\P_1,\P_2)$ of Example \ref{specex}. Since $\T_{\P_i}$ has no independence squares, $\mathscr{A}(\T_{\P_i}) = U(\T_{\P_i})$. Since $\T_{\P_1,\P_2} \cong \T_{\P_1} \interleave  \T_{\P_2}$, it follows that $\mathscr{A}(\T_{\P_1,\P_2})$ is isomorphic to the tensor product $U(\T_{\P_1})\otimes U(\T_{\P_2})$, which geometrically is a torus. 
\end{exa}

\section{Nondeterministic choice} \label{SecChoice}

The coproduct of $\ltimes$-transition systems models the nondeterministic sum of shared-variable systems. We show that the coproduct of the HDA models of two $\ltimes$-transition systems is an HDA model of their coproduct. 

\subsection{Coproduct of HDAs}

The coproduct of two HDAs $\A$ and $\B$ is the HDA $\A+\B$ where
\begin{itemize}
	\item $P_{\A+\B}$ is the precubical subset of $P_\A \otimes P_\B$ defined by \[(P_{\A + \B})_m = {(\{I_\A\} \times (P_\B)_m)} \cup ((P_\A)_m \times \{I_\B\}) \quad (m \geq 0);\]
	\item $I_{\A+\B} = (I_\A,I_\B)$;
	\item $F_{\A+\B} = (\{I_\A\} \times F_\B) \cup (F_\A \times \{I_\B\})$;
	\item $\Sigma_{\A+\B} = \Sigma_\A \amalg \Sigma_\B$; and 
	\item $\lambda_{\A+\B}(x,y) = \left\{ \begin{array}{ll}
	\lambda_\B(y), & (x,y) \in \{I_\A\}\times (P_\B)_1,\\
	\lambda_\A(x), & (x,y) \in (P_\A)_1\times \{I_\B\}.	
	\end{array}\right.$ 	
\end{itemize} 
If $\A$ and $\B$ are extensional, then so is $\A + \B$. Moreover, in this case $\A+\B$ is the coproduct of $\A$ and $\B$ in $\mathsf{HDA_{ext}}$.  

\subsection{Coproduct of \texorpdfstring{$\ltimes$}{}-transition systems}

The coproduct of two $\ltimes $-transition systems $\S$ and $\T$ is the $\ltimes $-transition system $\S + \T$ where $U(\S + \T) = U(\S) + U(\T)$ and $\alpha \ltimes_{\S \interleave \T} \beta$ if either 
\begin{enumerate}
	\item $\alpha, \beta \in \Sigma_\S$ and $\alpha\ltimes_\S \beta$; or
	\item $\alpha, \beta \in \Sigma_\T$ and $\alpha\ltimes_\T \beta$.
\end{enumerate}

\subsection{Nondeterministic sum of shared-variable systems}
Let $Var$ be a set of variables, and let $\eta_0 \in Ev(Var)$ be an initial evaluation. We say that a shared-variable system $(\P_1, \dots, \P_n)$ over $Var$ is the \emph{nondeterministic sum} with respect to $\eta_0$ of two shared-variable systems $(\P^1_1, \dots, \P^1_{n_1})$ and $(\P^2_1, \dots, \P^2_{n_2})$ over $Var$ if $\T_{\P_1, \dots, \P_n} \cong \T_{\P^1_1, \dots, \P^1_{n_1}} + \T_{\P^2_1, \dots, \P^2_{n_2}}$.

\begin{exa}
Consider the shared-variable system $(\P_1, \P_2)$ of Example \ref{specex} and modify the function $g$ of $\P_1$ and $\P_2$ by setting $g(0,x{\scriptstyle{++}},1) = \{0\}$. While this has no effect on $\T_{\P_1}$ and $\T_{\P_2}$, $\T_{\P_1,\P_2}$ is altered in that the state $(1,1,2)$ and the transitions involving it disappear and $\ltimes_{\T_{\P_1,\P_2}}$ becomes the empty relation. Therefore one now has $\T_{\P_1,\P_2} \cong \T_{\P_1} + \T_{\P_2}$, and so $(\P_1,\P_2)$ is the nondeterministic sum of the systems $(\P_1)$ and $(\P_2)$.
\end{exa}

\subsection{HDA models of coproducts} Let $\S$ and $\T$ be $\ltimes$-transition systems, and let $\A$ and $\B$ be HDA models of $\S$ and $\T$, respectively.

\begin{thm} \label{sum}
	$\A + \B$ is an HDA model of $\S + \T$.
\end{thm}

\begin{proof}
(HM1) $(\A + \B)_{\leq 1} = \A_{\leq 1} + \B_{\leq 1} = U(\S)+U(\T) = U(\S+\T)$.

(HM2) Let $(x,y) \in (P_{\A+\B})_2$. If $(x,y) \in \{ I_\S\} \times (P_\B)_2$, then \[\lambda_{\A+\B}(d^0_2(x,y))  =  \lambda_\B(d^0_2y) \ltimes_\T \lambda_\B(d^1_2y) = \lambda_{\A + \B}(d^1_2(x,y)).\] Hence $\lambda_{\A+\B}(d^0_2(x,y)) \ltimes_{\S+\T} \lambda_{\A + \B}(d^1_2(x,y))$. A similar argument shows that this also holds if $(x,y) \in (P_\A)_2 \times \{I_\T\}$. 

(HM3) Let $m\geq 2$ and $(a,b), (x,y) \in (P_{\A+\B})_m$ such that $d^k_r(a,b) = d^k_r(x,y)$ for all $r \in \{1, \dots, m\}$ and $k \in \{0,1\}$. Suppose first that $(a,b) \in \{ I_\S\} \times (P_\B)_m$. Then also $(x,y) \in \{ I_\S\} \times (P_\B)_m$. Indeed, if this was not the case, then we would have $(x,y) \in (P_\A)_m\times \{ I_\T\}$ and $(I_\S,d^0_1b) = d^0_1(a,b) = d^0_1(x,y) = (d^0_1x,I_\T)$. Since $\{I_\S\}\times (P_\B)_{m-1}$ and $(P_\A)_{m-1}\times \{I_\T\}$ are disjoint subsets of $(P_{\A+\B})_{m-1}$, this is impossible. Thus $(x,y) \in \{ I_\S\} \times (P_\B)_m$. It follows that $d^k_rb = d^k_ry$ for all $r \in \{1, \dots, m\}$ and $k \in \{0,1\}$. By property HM3 for $\B$, $b=y$. Thus $(a,b) = (x,y)$. A similar argument shows that $(a,b) = (x,y)$ if $(a,b) \in (P_\A)_m \times \{ I_\T\}$. 

(HM4) Let $\C$ be an HDA that satisfies HM1-HM3 for $\S + \T$ and contains $\A+\B$ as a subautomaton. We show by induction that for all $r\geq 1$, $\C_{\leq r} = (\A + \B)_{\leq r}$. We have $\C_{\leq 1} = U(\S+\T) = (\A +\B)_{\leq 1}$. 
Let $r \geq 2$, and suppose that $\C_{\leq r-1} = (\A + \B)_{\leq r-1}$. Let $c \in (P_\C)_r$. We show that $c \in (P_{\A +\B})_r$. By the induction hypothesis, we have $d^k_ic \in (P_{\A +\B})_{r-1}$ for all $k \in \{0,1\}$ and $i \in \{1, \dots, r\}$. 

We show first that either for all $k$ and $i$, $d^k_ic \in {\{I_\S\}\times (P_\B)_{r-1}}$  or for all $k$ and $i$, $d^k_ic \in {(P_\A)_{r-1}\times \{I_\T\}}$. If $r=2$, then this follows from the fact that, by property HM2, $\lambda_\C(d^0_2c) \ltimes_{\S+\T} \lambda_\C(d^0_1c)$. Indeed, this implies that either all $\lambda_\C(d^k_ic) \in \Sigma_\S$ or all $\lambda_\C(d^k_ic) \in \Sigma_\T$, which is only possible if either all $d^k_ic \in {\{I_\S\}\times (P_\B)_{r-1}}$ or all $d^k_ic \in {(P_\A)_{r-1}\times \{I_\T\}}$. Suppose that $r > 2$. If $d^k_1c \in {\{I_\S\}\times (P_\B)_{r-1}}$, then for all $j \in \{2, \dots, r\}$ and $l \in\{0,1\}$, $d^k_1d^l_jc = d^l_{j-1}d^k_1c \in \{I_\S\} \times (P_\B)_{r-2}$ and therefore $d^l_jc \in \{I_\S\} \times (P_\B)_{r-1}$. Similarly, if $d^k_1c \in {(P_\A)_{r-1}\times \{I_\T\}}$,  then for all $j \in \{2, \dots, r\}$ and $l \in\{0,1\}$, $d^l_jc \in (P_\A)_{r-1}\times \{I_\T\}$. It follows that either $d^k_ic \in {\{I_\S\}\times (P_\B)_{r-1}}$ for all $k, i$ or $d^k_ic \in {(P_\A)_{r-1}\times \{I_\T\}}$ for all $k, i$.  

Suppose that all $d^k_ic \in {\{I_\S\}\times (P_\B)_{r-1}}$. Write $d^k_ic = (I_\S, y^k_i)$. For $k,l \in\{0,1\}$ and $1 \leq i < j \leq r$, $(I_\S, d^k_iy^l_j) = d^k_id^l_jc = d^l_{j-1}d^k_ic = (I_\S,d^l_{j-1}y^k_i)$ and hence $d^k_iy^l_j = d^l_{j-1}y^k_i$. If $r=2$, we also have $\lambda_\B(y^0_i) = \lambda_\C(d^0_ic) = \lambda_\C(d^1_ic) = \lambda_\B(y^1_i)$ for all $i\in\{1,2\}$. Moreover, $\lambda_\C(d^0_2c) \ltimes_{\S+\T} \lambda_\C(d^0_1c)$ and therefore $\lambda_\B(y^0_2) \ltimes_\T  \lambda_\B(y^0_1)$. By Proposition \ref{filling}, there exists an element $y \in (P_\B)_r$ such that $d^k_iy = y^k_i$ for all $k,i$. Hence for all $k,i$, $d^k_ic = (I_\S, y^k_i) = d^k_i(I_\S,y)$. By property HM3, it follows that $c = (I_\S,y) \in (P_{\A+\B})_r$. An analogous argument shows that $c \in (P_{\A+\B})_r$ if all $d^k_ic \in (P_\A)_{r-1}\times \{I_\T\}$. 
\end{proof}

\begin{rem}
	By the universal property of the coproduct, there exists a unique morphism of HDAs $\mathscr{A}(\S) + \mathscr{A}(\T) \to \mathscr{A}(\S+\T)$ that restricts to the (obvious) morphisms $\mathscr{A}(\S \to \S+\T)$ and $\mathscr{A}(\T \to \S+\T)$. The morphism $(\mathscr{A}(\S) + \mathscr{A}(\T))_{\leq 1} \to \mathscr{A}(\S+\T)_{\leq 1}$ is the identity on $U(\S+\T)$. By Theorem \ref{sum} and Corollary \ref{isomodel}, it follows that $\mathscr{A}(\S) + \mathscr{A}(\T) \to \mathscr{A}(\S+\T)$ is an isomorphism. Therefore the functor $\mathscr{A}$ preserves (binary) coproducts. 
\end{rem}

\section{Concluding remarks} \label{SecFinal}

We have introduced $\ltimes$-transition systems to model shared-variable programs and presented a method to construct higher-dimensional automata from them. We have established that the construction  is compositional in the sense that interleaving and nondeterministic choice are modeled by the tensor product and the coproduct of HDAs, respectively.

The squares and higher-dimensional cubes in an HDA represent independence of actions. The difference between a filled and an empty cube is a topological one, and this strongly suggests to use topological methods to analyze the independence structure of HDAs. An approach to this based on homology is described in \cite{labels}. In the HDA models considered in the present paper the independence concept for actions is induced by a binary relation. Since this does not extend to sequences of actions, these HDAs have usually nonetheless a nontrivial higher independence structure and topology.

The main objective of this paper was to provide a framework in which one can actually prove the intuitively obvious fact that the tensor product of HDAs models the parallel composition of independent concurrent systems. Unfortunately, the tensor product of HDAs has the undesirable property not to be symmetric, i.e., normally, $\A \otimes \B \not\cong \B \otimes \A$. One way to overcome this problem is to consider a larger class of morphisms in the category of HDAs. In previous work this has already been advocated as a means to compare HDAs of different refinement levels \cite{weakmor, topabs}. It turns out that the \emph{weak morphisms} considered in those articles are still too rigid to make the tensor product symmetric. The more flexible \emph{cubical dimaps} introduced in \cite{labels} seem, however, to be a suitable choice of morphisms for a symmetric monoidal category of higher-dimensional automata. This issue will be a subject of future work.


\input{refs.bbl}

\end{document}

%% file: refs.bbl
\providecommand{\bysame}{\leavevmode\hbox to3em{\hrulefill}\thinspace}
\providecommand{\MR}{\relax\ifhmode\unskip\space\fi MR }
\providecommand{\MRhref}[2]{%
  \href{http://www.ams.org/mathscinet-getitem?mr=#1}{#2}
}
\providecommand{\href}[2]{#2}